\documentclass[nohyperref]{article}

\usepackage[noend]{algpseudocode}
\usepackage[accepted]{icml2022}

\usepackage{times}

\usepackage{soul}
\usepackage{url}
\usepackage[hidelinks]{hyperref}
\usepackage[utf8]{inputenc}
\usepackage[small]{caption}
\usepackage{graphicx}
\usepackage{amsmath}
\usepackage{booktabs}
\usepackage{amsfonts}
\usepackage{amssymb}
\usepackage{multirow}

\usepackage{xcolor}
\usepackage{tabularray}

\usepackage{placeins}
\usepackage{float, subcaption}
\usepackage{bm}
\usepackage{amsthm}

	\usepackage{bbm} 
	\let\emptyset\varnothing 

\usepackage{mathtools}
\usepackage{bbm}
\usepackage{tikz}
\usepackage{tikz-3dplot}
\usepackage{pgfplots}
\usepgfplotslibrary{ternary}
\usetikzlibrary{calc,shapes}
\usepackage{dsfont}
\usepackage{enumitem}
\usepackage{thm-restate}
\usepackage{multirow}
\usepackage{newfloat}
\usepackage{wrapfig}
\usepackage{mleftright}
\usepackage[capitalise,noabbrev]{cleveref}

\usepackage{pifont}
\usepackage[normalem]{ulem}

\setlist{nosep,leftmargin=*}

\theoremstyle{plain}
\newtheorem{theorem}{Theorem}[section]
\newtheorem{proposition}[theorem]{Proposition}

\theoremstyle{definition}
\newtheorem{definition}[theorem]{Definition}

\theoremstyle{remark}

\usepackage[textsize=tiny]{todonotes}
\usepackage{xspace}

\icmltitlerunning{Team-Public-Information Representation for Adversarial Team Games}

\begin{document}

\twocolumn[
\icmltitle{A Marriage between Adversarial Team Games and 2-player Games: \\ Enabling Abstractions, No-regret Learning, and Subgame Solving}



\icmlsetsymbol{equal}{*}

\begin{icmlauthorlist}
\icmlauthor{Luca Carminati}{polimi}
\icmlauthor{Federico Cacciamani}{polimi}
\icmlauthor{Marco Ciccone}{polito}
\icmlauthor{Nicola Gatti}{polimi}
\end{icmlauthorlist}

\icmlaffiliation{polimi}{Politecnico di Milano}
\icmlaffiliation{polito}{Politecnico di Torino}

\icmlcorrespondingauthor{Luca Carminati}{luca.carminati@polimi.it}

\icmlkeywords{Algorithmic Game Theory, Cooperative AI, Machine Learning, ICML}

\vskip 0.3in
]



\printAffiliationsAndNotice{}  
%

\makeatletter
\DeclareRobustCommand\onedot{\futurelet\@let@token\@onedot}
\def\@onedot{\ifx\@let@token.\else.\null\fi\xspace}

\def\eg{\emph{e.g}\onedot} \def\Eg{\emph{E.g}\onedot}
\def\ie{\emph{i.e}\onedot} \def\Ie{\emph{I.e}\onedot}
\def\cf{\emph{c.f}\onedot} \def\Cf{\emph{C.f}\onedot}
\def\etc{\emph{etc}\onedot} \def\vs{\emph{vs}\onedot}
\def\wrt{w.r.t\onedot} \def\dof{d.o.f\onedot}
\def\etal{\emph{et al}\onedot}
\makeatother


\newcommand{\argmin}{arg\min}
\newcommand{\argmax}{arg\max}
\newcommand{\Reals}{\mathbb{R}}
\newcommand{\team}{\mathcal{T}}
\newcommand{\Z}{\mathcal{Z}}
\newcommand{\St}{\mathcal{S}}
\newcommand{\N}{\mathcal{N}}

\newcommand{\I}{\mathcal{I}}
\newcommand{\algoname}[1]{\textnormal{\textsc{#1}}}

\definecolor{mygreen}{rgb}{0.0, 0.5, 0.0}
\newcommand{\nb}[3]{{\colorbox{#2}{\bfseries\sffamily\scriptsize\textcolor{white}{#1}}}{\textcolor{#2}{\sf\small\textit{#3}}}}

\newcommand{\lu}[1]{\nb{Luca}{blue}{#1}}
\newcommand{\ma}[1]{\nb{Marco}{red}{#1}}
\newcommand{\fe}[1]{\nb{Federico}{mygreen}{#1}}
\newcommand{\nc}[1]{\nb{Nicola}{violet}{#1}}

\setlength{\belowdisplayskip}{3pt} \setlength{\belowdisplayshortskip}{3pt}
\setlength{\abovedisplayskip}{3pt} \setlength{\abovedisplayshortskip}{3pt}
\setlength{\itemsep}{0pt}
\setlength{\topsep}{1pt}
\setlength{\parskip}{5pt}

\begin{abstract}
    \emph{Ex ante} correlation is becoming the mainstream approach for \emph{sequential adversarial team games}, where a team of players  faces  another team in a zero-sum game. It is known that  team members' asymmetric information makes both equilibrium computation \textsf{APX}-hard and team's strategies not directly representable on the game tree. This latter issue prevents the adoption of successful tools for huge 2-player zero-sum games such as, \emph{e.g.}, abstractions, no-regret learning,  and subgame solving.  This work shows that we can recover from this weakness by bridging the gap between sequential adversarial team games and 2-player games. In particular, we propose a new, suitable game representation that we call \emph{team-public-information}, in which a team is represented as a single coordinator who only knows information common to the whole team  and prescribes to each member an action for any possible private state. The resulting representation is highly \emph{explainable}, being a 2-player tree in which the team's strategies are behavioral with a direct interpretation and more expressive than the original extensive form when designing abstractions. Furthermore, we prove payoff equivalence of our representation, and we provide techniques that, starting directly from the extensive form, generate dramatically more compact representations without information loss.
    Finally, we experimentally evaluate our techniques when applied to a standard testbed, comparing their performance with the current state of the art.
\end{abstract}
\section{Introduction}\label{sec:intro}%
Research efforts on imperfect-information games customarily focus on 2-player zero-sum games (\textit{``2p0s''} games from here on), in which two players act receiving opposite payoffs. In this setting, superhuman performances have been achieved in real-world instances, such as \textit{Poker Hold'em} \cite{BrownS17,Brown885,Moravck2017DeepStackEA} and \textit{Starcraft II}  \cite{Vinyals2019GrandmasterLI}.
%
%
The successful approach for 2p0s games is generally based on the generation of a game abstraction used \emph{offline} to find a blueprint strategy which is refined \emph{online} during the play.
%

In our work, we focus on sequential \emph{adversarial team games} in which a team of 2 (or more) players cooperates against a common adversary or team of adversaries. In particular, we focus on  \textit{ex ante coordination}, in which the team members agree on a common strategy beforehand and commit to playing it during the game without communicating any further. The team members share the same payoffs and coordinate against an adversary having opposite payoffs, in face of private information given separately to each team member. Examples include collusion in poker games, bidding in the game of Bridge, and a team of drones acting against an intruder.  \citeauthor{Celli2018ComputationalRF}~(\citeyear{Celli2018ComputationalRF}) show that the computation of a solution, called Team Maxmin Equilibrium with Correlation (TMEcor), is \textsf{APX}-hard.  Furthermore, team members' asymmetric information makes a team equivalent to a single-player without \emph{perfect recall} and therefore, as showed by \citeauthor{kuhn1953}~(\citeyear{kuhn1953}), \emph{behavioral} strategies defined on the game tree and \emph{normal-form} strategies are not realization equivalent. In particular, normal-form strategies may lead to arbitrarily better outcomes than behavioral strategies. However, this comes at the cost of an exponential explosion of the strategy space and the impossibility to use tools for huge 2p0s games as normal-form strategies are not directly representable on game trees.


\textbf{Related Work.} 
To the best of our knowledge, \citet{Celli2018ComputationalRF} are the first to compute the TMEcor of an adversarial team game by proposing the Hybrid Column Generation (HCG) algorithm. At each iteration, HCG exploits a Linear Program (LP) to compute a max-min solution and then an Integer LP (ILP) to find the team's best response to be added to the LP at the next iteration. 
Successively, \citet{Farina2018ExAC}~propose a variant of HCG, called Fictitious Team Play (FTP), in which the LP computing the max-min strategy is replaced by a step of the Fictitious Play algorithm~\cite{BrownFP}. 
Later, \citet{DBLP:conf/aaai/00010C21}, \citet{DBLP:conf/icml/0004020}, 
\citet{DBLP:conf/aaai/Zhang020}, \citet{Farina2021ConnectingOE} propose more efficient flavours of HCG and FTP algorithms.
Among the above algorithms, the Faster Column Generation (FCG) algorithm~\cite{Farina2021ConnectingOE} provides the best empirical performance.
%
The rationale behind this class of approaches is to incrementally expand the LP strategy space to guess the actions in the equilibrium support without necessarily enumerating an excessively large portion of the space. The main weakness of this approach is the necessity to solve an ILP, which severely limits its scalability to large game instances even for the evaluation of the exploitability of a suboptimal solution.
A recent alternative is proposed by \citet{zhang2022team}. The authors provide a generalization of the sequence form which, thanks to a suitable tree decomposition of the constraints, allows the description of a team's strategy space by a polytope. Thus, a TMEcor can be found by linear programming. This approach outperforms FCG with instances in which the degree of private information is limited. The idea to provide a convex representation of the strategy space adopted by \citet{zhang2022team} is closely related to ours. The main differences reside in a better interpretability of our representation, together with the possibility to adopt abstractions~\citet{Sandholm2015AbstractionFS,DBLP:conf/aaai/GilpinSS07}, no-regret learning~\cite{Zinkevich2007RegretMI,celli2020no}, and subgame solving~\cite{Brown2018DepthLimitedSF,brown2017safe}.

We also mention Multi-Agent Reinforcement Learning (MARL) approaches proposed by~\citet{Celli2019CoordinationIA} and \citet{Cacciamani2021MultiAgentCI}.
These algorithms rely on implicit abstractions yielded by deep reinforcement learning to reduce the complexity of the problem. However, these approaches provide theoretical guarantees only in games in which team members have symmetric observability over other players' actions (chance included). 

\textbf{Original Contributions.} 
As a preliminary step of our work, we first enrich the canonical extensive-form representation to capture information about public team members' observations. 
We call it extensive-form game \emph{with visibility} (vEFG). %
Exploiting this representation,  we provide an algorithmic procedure, called \algoname{PublicTeamConversion}, to convert an  instance of adversarial team games into a 2p0s game, where a team is represented as a single coordinator who only knows information common to all team members and prescribes to each member an action for any possible private state.
We formally prove that a Nash equilibrium of the converted game corresponds to a TMEcor in the original game and \emph{vice versa}, thus enabling, for the first time, to the best of our knowledge, the adoption of techniques for 2p0s games to adversarial team games.

Differently from the representations previously proposed in the state of the art, \emph{e.g.}, that by~\citet{zhang2022team}, our representation is highly \textit{explainable}, since the team's strategies are behavioral over the game tree with a direct interpretation. More precisely, \emph{the coordination prescriptions sent to the team members in each public state can be interpreted as shared team conventions}. 
Remarkably, our representation also extends to adversarial settings the research line previously developed by \citet{Nayyar2013DecentralizedSC} and applied  to  cooperative games by \citet{Foerster2019BayesianAD} and \citet{Sokota2021SolvingCG}, thus bridging the two approaches.

Furthermore, we show that our representation is more expressive than the extensive form as state/action abstractions applied to the extensive-form game can be captured by our representation, while the reverse does not hold. 
More importantly, the direct interpretability of our representation allows the design of techniques to prune and abstract the trees. In particular, we show that our techniques return a game representation with a size smaller than that generated by \citet{zhang2022team}, while guaranteeing explainability. Finally, we empirically evaluate the performance of no-regret algorithms applied to our representations.
\section{Preliminaries}\label{sec:prelim}
We introduce the basic concepts and definitions used throughout this work. For more details, we point an interested reader to \citeauthor{shoham-leyton}~(\citeyear{shoham-leyton}).
\paragraph{Extensive-form Games and Adversarial Team Games}

The basic model for sequential interactions among a set~$\mathcal{N}$ of $N$ players with private information is the \emph{Extensive-Form Game with imperfect information} (EFG). 
An EFG is a tuple $(\mathcal N,  \mathcal H, \mathcal Z, \iota, \mathcal A, A, \mathcal I, \{u_p\}_{p\in\mathcal N})$ defining a tree where the set of nodes is denoted by $\mathcal{H}$ and the set of leaves (a.k.a.~terminal nodes) is denoted by $\mathcal{Z}\subseteq\mathcal{H}$. 
The player acting at node $h\in\mathcal{H}$ is returned by function $\iota(h) \in \mathcal{N}$.
Set $\mathcal{A}= \cup_{p \in \mathcal{N}}\mathcal{A}_p$ contains all the possible actions, where $\mathcal{A}_p$ is the set of actions available to player~$p \in \mathcal{N}$. Given a node $h$, the set of available actions at $h$ is $A(h)$. We also refer to a node $h$ as a \emph{history}, meaning the sequence of all the actions from the root to node $h$. 
Let $u_p: \mathcal{Z}\to\Reals$ be the \textit{payoff function} of player~$p$ mapping every terminal node to a utility value.
%
%
In order to account for imperfect information, we use \emph{information sets} (for brevity, infosets).
An infoset (also called private state) $I\subseteq\mathcal{H}\setminus \mathcal{Z}$ is a partition of the player~$p$'s nodes that are indistinguishable to~$p$. 
We denote the set of player $p$'s infosets as $\mathcal{I}_p$ and the set of all information sets as $\mathcal I = \cup_{p\in\mathcal N} \mathcal I_p$.
With notation overload, we use $\iota(I)$ and $A(I)$ in place of $\iota(h)$ and $A(h)$ where $h \in I$ and we denote as $I(h)$ the infoset corresponding to node $h$, for any $h\in\mathcal H$.
%
%

We focus on Adversarial Team Games (ATGs).
An ATG is an $N$-player EFG in which a \emph{team} of players $\mathcal{T}\subseteq\mathcal{N}$ plays against an opponent $o$ (or a team of players). 
If chance player $c$ is present, we enrich the set of players with it.
Thus, $\mathcal{N} = \mathcal{T}\cup\left\{ o\right\}\cup\left\{ c\right\}$. 
A team is a set of players sharing the same utility function.
Formally, $\forall p \in \mathcal{T}$, $u_p = u_\mathcal{T}$ for some function $u_\mathcal{T}$. 
We restrict our analysis to zero-sum ATGs, in which  $u_{\mathcal{T}} = -u_o$. 
Note that since chance $c$ is a non-strategic player, its payoff is not defined.
For an EFG, a \emph{deterministic timing} is a labeling of the nodes in $\mathcal{H}$ with natural numbers such that the label of any node is strictly higher than the label of its parent. 
A deterministic timing is \emph{exact} if all nodes in the same information set have the same label, and the game is called \emph{timeable}. 
Furthermore, an EFG is \emph{1-timeable}, when admitting an exact timing where the difference between the labels of the nodes and their parents is one.
Furthermore, we focus on \emph{perfect recall} games, in which no player forgets information.
Exploiting the property of 1-timeability, we can define an ordering between different nodes. In particular, for two nodes $h,h'\in \mathcal H$ we say that $h$ precedes $h'$ (denoted as $h\preccurlyeq h'$) if the label assigned to $h$ is smaller than the label assigned to $h'$ and in the path from the root of the game tree to $h'$, node $h$ is encountered. With a slight abuse of notation, for two infosets $I,J\in \mathcal I$, we write that $I\preccurlyeq J$ if  there exists $h \in I, h'\in J$ such that $h\preccurlyeq h'$. 
In addition, given an infoset $I$, the set of team members that will play in some infoset following $I$ is denoted with $\mathcal{T}_I := \{p\in\mathcal T \mid  \exists J \in \mathcal I_p\, \text{s.t. } I\preccurlyeq J\}$.

\paragraph{Strategies and Nash Equilibrium}
Game theory provides various strategy representations in EFGs. 
A \emph{behavioral strategy} $\sigma_p: \mathcal{I}_p\to\Delta^{|A(I)|}$ is a function that maps each infoset $h$ to a probability distribution over available actions $A(h)$. 
A \emph{normal-form plan} (or \emph{pure strategy}) $\pi_p\in\Pi_p:= \bigtimes_{I\in\mathcal{I}_p}A(I)$ is a tuple specifying one action for each infoset, while a \emph{normal-form strategy} $\mu_p\in\Delta^{|\Pi_p|}$ is a probability distribution over normal-form plans.
\citeauthor{kuhn1953}~(\citeyear{kuhn1953}) show that behavioral and normal-form strategies are equivalent in perfect-recall games, while this does not hold with imperfect recallness where normal-form strategies are (usually) more expressive than behavioral.
A \emph{reduced normal-form strategy} $\mu^\star_p$ is obtained from a normal-form strategy $\mu_p$ by aggregating plans distinguished by action played in unreachable nodes.
With a slight abuse of notation, $\forall p \in\mathcal{N}$, we denote with $\sigma_p[z]$ (respectively $\mu_p[z]$) the probability of reaching terminal node $z\in\Z$ when following strategy $\sigma_p$ (resp. $\mu_p$). 
A strategy profile is a tuple associating a strategy to each player in the game. 
We denote normal-form strategy profiles with $\boldsymbol{\mu}$ and behavioral strategy profiles with $\boldsymbol{\sigma}$. 
Given a strategy profile $\boldsymbol{\mu}$, we denote with $\mu_p$ the strategy of player $p\in\N$ and with $\boldsymbol{\mu}_{-p}$ the strategies of all the other players. 
With an abuse of notation, the expected utility for player $p$ when she plays strategy $\mu_p$ and all the other players play strategy $\boldsymbol{\mu}_{-p}$ is $u_p(\mu_p, \boldsymbol{\mu}_{-p})$.
Furthermore, we define the \emph{best response} of player $p$ to strategy profile $\boldsymbol{\mu}_{-p}$ as the strategy that maximizes player $p$'s utility against strategy $\boldsymbol{\mu}_{-p}$.
Formally, $\mathsf{BR}_p(\boldsymbol{\mu}_{-p}) := \arg\max_{\mu} u_p(\mu, \boldsymbol{\mu}_{-p})$.
A strategy profile $\boldsymbol{\mu}$ is a Nash Equilibrium (NE) if it is stable with respect to unilateral deviations of a single player. 
Formally, $\boldsymbol{\mu}$ is a NE if and only if $\forall p \in \mathcal{N}$, $\mu_p \in \mathsf{BR}_p(\boldsymbol{\mu}_{-p})$.
\paragraph{\emph{Ex ante} Coordination in ATGs}
\citeauthor{Basilico2016teammaxmin}~(\citeyear{Basilico2016teammaxmin}) show that the team's expected payoff in a Nash equilibrium can be arbitrarily smaller than the payoff in a \emph{Team Maxmin Equilibrium}, introduced by \citeauthor{VonStengelKoller97}~(\citeyear{VonStengelKoller97}), which in its turn can be arbitrarily smaller than the payoff in a Team Maxmin Equilibrium with Correlation strategies.
%
%
%
The TMEcor can be computed through a LP formulated over the joint normal-form plans of the team players: 
\begin{equation}\label{eq:tmecor}
\begin{array}{l}\displaystyle
\max_{\mu_\team}\min_{\mu_o} \sum_{z\in\Z}\hspace{.4cm} \mu_\team[z]\,\mu_o[z]\,u_\team(z)\\
\hspace{.3cm}\textnormal{s.t. }\hspace{.5cm}\mu_\team\in\Delta(\bigtimes\limits_{p\in\team}\Pi_p)\\[4mm]
\hspace{1.3cm}\mu_o\in\Delta(\Pi_o).
\end{array}
\end{equation}
The team strategy space $\bigtimes_{p\in\team} \Pi_p$ can grow exponentially in the size of the game tree, thus making Problem~\eqref{eq:tmecor} unaffordable in practice except for toy games.
\section{Extensive-Form Games with Visibility Representation}\label{sec:vefg}
We introduce the concept of \emph{Extensive-Form Game with visibility}. This representation allows us to explicitly capture the information common to a set of players (\emph{e.g.}, team members) and to extend the notion of infoset accordingly. 
%

\paragraph{Public Function.} We first introduce a function $Pub_p:\mathcal{A}\to\{\text{obs, unobs}\}$, $\forall p\in\mathcal{N}$, specifying whether action $a\in\mathcal{A}$ is \emph{observable} or \emph{unobservable}, respectively, by a single player~$p$ when $a$ is played by another player.
Note that our definition of $Pub_p$ does not depend on the nodes in which player~$p$ plays, and therefore it cannot capture potential imperfect recallness in which $p$ forgets actions observed before.
Trivially, the information structure of every perfect-recall game is induced by some $\{Pub_{p}\}_{p \in \mathcal{N}}$:
\begin{proposition}
Any pair of histories $h,h'$ of player~$p$ belong to the same infoset when the actions $a$ in $h$ observable by $p$ and the actions $a'$ in $h'$ observable by $p$ are the same, formally, when ${(a)_{a\in h : Pub_{p}(a)=\mathrm{obs}} = (a')_{a'\in h' : Pub_{p}(a')=\mathrm{obs}}}$. 
\end{proposition}

With notation overload, the definition of function $Pub$ can be extended to a set of players $\mathcal{P}$ (\emph{e.g.}, a team) as $Pub_\mathcal{P}: \mathcal{A}\to\{\mathrm{pub}, \mathrm{priv}, \mathrm{hidden}\}$, $\forall \mathcal P\subseteq\mathcal N$:
\begin{align*}
& Pub_{\mathcal{P}}(a) = \mathrm{pub} \iff \forall p \in \mathcal{P}:Pub_p(a)=\mathrm{obs};\\
& Pub_{\mathcal{P}}(a) = \mathrm{hidden} \iff \forall p \in \mathcal{P}:Pub_p(a) = \mathrm{unobs};\\
& Pub_{\mathcal{P}}(a) = \mathrm{priv} \text{ otherwise}.
\end{align*}
Informally, action~$a$ is called $\mathrm{pub}$ for a set of players $\mathcal{P}$, when it is observable by all the players of that set; $\mathrm{hidden}$, when it is not observable by all these players (notice that in this case $a$ is played by a player not belonging to $\mathcal{P}$); and $\mathrm{priv}$ when some player(s) in $\mathcal{P}$ can observe it, while some other player(s) in $\mathcal{P}$ cannot.
%
%
%
Finally, we can extend the standard definition of Extensive-Form game:
\begin{definition}[Extensive-Form Game with Visibility]
An Extensive-Form Game with Visibility (vEFG) is a tuple defined as $(\mathcal N, \mathcal H, \mathcal Z, \iota, \mathcal A, A,\mathcal I, \{u_p\}_{p\in\mathcal N}, \{Pub_p\}_{p \in \mathcal{N}})$ where $\mathcal{I}$ is induced by $\{Pub_p\}_{p \in \mathcal{N}}$ as discussed above, and therefore every player is with perfect recall.
\end{definition}


\subsection{Beyond Infoset: Public State}
 By means of $Pub_{\mathcal{P}}$, we can introduce the notion of \textbf{public state} for a set of players $\mathcal{P} \subseteq \mathcal{N}$, which extends the notion of infoset to a set of players. 
\begin{definition}[Public State]
A public state $S$ is a subset of nodes $\mathcal{H}$ such that any pair of histories $h,h'$ of potentially different players in $\mathcal{P}$ belong to $S$ when the actions $a$ in $h$ that are public for $\mathcal{P}$ and the actions $a'$ in $h'$ that are public for $\mathcal{P}$ are the same, formally, ${(a)_{a\in h : Pub_{\mathcal{P}}(a)=\mathrm{pub}} = (a')_{a'\in h' : Pub_{\mathcal{P}}(a')=\mathrm{pub}}}$. 
\end{definition}
In other words, two histories belong to the same public state if they share the same public actions and differ only for their private actions. We call $\mathcal{S}$ the set of all public states.
It can be easily seen that, if the node $h$ of an infoset $I$ belongs to a public state $S$, then $S$ also contains all the other nodes of $I$, and the notion of public state reduces to the notion of infoset when $\mathcal{P}$ is composed of a single player. In principle, a public state can contain multiple infosets, that can be of the same player and/or of different players in $\mathcal{P}$. In the case in which $\mathcal{P}$ is a team of players, we call a public state for $\mathcal{P}$ as a \emph{team-public infoset}. With an abuse of notation, for any set of players $\mathcal{P}\subseteq N$, we denote as $\mathcal{S}_{\mathcal{P}}(h)$ the set of all infosets belonging to players in $\mathcal{P}$ that are in the same public state as node $h$.


\subsection{Public-turn-taking Games} 
We focus on a class of games, called \textbf{public-turn-taking}, in which every player knows, at every infoset she plays, the sequence of players acted from the root to that infoset. 
This property refines 1-timeability as it requires that, in addition to the length of the history, even the sequence of players is common knowledge.
In public-turn-taking games, the public states have a specific structure that is central in our results, allowing the translation of an ATG as a 2p0s game. More precisely, every public state is composed of nodes of a single player whose  histories have the same length.

\begin{definition}[Public turn-taking property]
A vEFG is public turn-taking if:
\begin{equation*}
\forall I \in \mathcal I, \forall h,h' \in I : (\iota(g))_{g \sqsubseteq h} = (\iota (g'))_{g' \sqsubseteq h'}.
\end{equation*}
\end{definition}
Interestingly, we can show that, given an extensive-form game satisfying perfect recallness and timeability, we can generate a strategically equivalent game satisfying public-turn-taking property, whose size is polynomially upper bounded in the size of the original game (proofs omitted in the main paper are in Appendix~\ref{app:proofs}).
\begin{restatable}[Transformation into a public-turn-taking game]{theorem}{sizepublicturn}
Given any timeable vEFG with players $\mathcal{N}$ and nodes $\mathcal{H}$, there is a strategically equivalent (admitting the same reduced normal form) public-turn-taking vEFG whose size is $O(\,|\mathcal N| \,\, |\mathcal H|^2\,)$.
\label{le:turnTaking}
\end{restatable}

\subsection{Completely inflated games} 
Another important class of team games for the \textit{ex ante coordination} scenario is called \textbf{completely inflated games}. In this class of games, every team member knows the exact action played by another team member at any information set. This property allows us to explicitly represent that teammates share their strategies before starting the game. 
\begin{definition}[Completely inflated vEFG \citep{Kaneko1995BehaviorSM}]
A vEFG $\mathcal G$ is completely inflated with respect to a team of players $\mathcal T$ if:
\begin{align}
Pub_\mathcal{T}(a) = \mathrm{pub}\ \forall a \in \mathcal A_p \forall p \in \mathcal T
\end{align}
\end{definition}

In the following, we focus on completely inflated vEFGs for the team $\mathcal T$. This can be ensured for a generic vEFG by modifying the function $Pub_t(a)$ in such a way that $Pub_t(a) = \mathrm{obs}\ \forall t \in \mathcal T \ \forall a \in \mathcal A_{t'} \forall t' \in \mathcal T$.

\section{Team-Public-Information Conversion Algorithm}\label{sec:conversion}

\subsection{Conversion Procedure}
We present the algorithmic procedure to convert an ATG into a 2p0s game, denoted as \emph{Team-Public-Information} (TPI) game, in which a coordinator player takes the strategic decision on behalf of the team.
The pseudo-code is provided in Algorithm~\ref{alg:PubTeam}.
\begin{definition}[Team-Public-Information game]
Given a completely inflated vEFG $\mathcal G$ that satisfies the public turn-taking property, the corresponding TPI game $\mathcal G'$\footnote{Superscript $'$ denotes the elements of the converted game.} is defined as the output of the function \algoname{ConvertGame}$(\mathcal G)$ described in Algorithm~\ref{alg:PubTeam}.
\end{definition}
\begin{algorithm}[!htb]
\caption{Team-Public-Information Conversion}\label{alg:PubTeam}
\begin{algorithmic}[1]
\small
	\Function {ConvertGame}{$\mathcal G$}
	    \Statex \hfill$\triangleright$~{\scriptsize $\mathcal G=(\mathcal N, \mathcal H, \mathcal Z, \iota, \mathcal A, A,\mathcal I, \{u_p\}_{p\in\mathcal N}, \{Pub_p\}_{p \in \mathcal{N}})$}
		\State initialize $\mathcal G'$ new game
		\Statex \hfill$\triangleright$~{\scriptsize $\mathcal G'=(\mathcal N', \mathcal H', \mathcal Z', \iota', \mathcal A', A',\mathcal I', \{u_p'\}_{p\in\mathcal N'}, \{Pub'_p\}_{p\in\mathcal N'})$}
		\State $\mathcal N' \gets \{t,o,c\}$
		\State $h'_{\emptyset} \gets$ \Call{PubTeamConv}{$h_{\emptyset}, \mathcal G, \mathcal G'$} \Comment{new game root}
		\State \Return{$\mathcal G'$}
	\EndFunction
	\item[] 
	\Function{PubTeamConv}{$h$, $\mathcal G$, $\mathcal G'$}
		\State initialize $h' \in \mathcal H'$
  		\If {$h \in \mathcal Z$} \Comment{terminal node}
  			\State $\Z' \gets \Z' \cup\{h'\} $
  			\State $u'_p(h') \gets u_p(h) \quad \forall p \in \mathcal N$
  			\State $u_t'(h')\gets \sum_{p\in\mathcal T} u_p(h)$
  			\State $u_o'(h') \gets -u_t'(h')$
  		\ElsIf {$\mathcal \iota(h) \in \{o, c\}$} \Comment{opponent or chance}
  			\State $\iota'(h') \gets \iota(h)$
  			\State $A'(h') \gets  A(h)$
  			\If {$\iota(h) = c$}
  				\State $\sigma_c'(h') = \sigma_c(h)$
  			\EndIf
  			\For {$a' \in A'(h')$}
  				\State $Pub_t'(a') \gets$ $\mathrm{obs}$ \textbf{if} $Pub_{\mathcal T}(a') = \mathrm{pub}$ \textbf{else} $\mathrm{unobs}$
  				\State $Pub_o'(a') \gets Pub_o(a')$
  				\State $h'a' \gets$ \Call{PubTeamConv}{$ha'$, $\mathcal G$, $\mathcal G'$}
  			\EndFor 
  		\Else \Comment{team member}
  			\State $\mathcal \iota'(h') = t$
  			\State $I\gets I(h)$
  			\State $A'(h') \gets \bigtimes_{J \in \mathcal  S_{\mathcal T_I}(h)} A(J)$ \label{lst:line:prescription} \Comment{prescriptions}
  			\For {$\Gamma' \in A'(h')$}
  				\State $Pub_t'(\Gamma') \gets \mathrm{seen}, Pub_o'(\Gamma') \gets \mathrm{unseen}$
  				\State $a' \gets \Gamma'[I(h)]$ \Comment{extract chosen action}
  				\State initialize $h'' \in \mathcal H'$
  				\State $A'(h'') \gets \{a'\}$
 				\State $\iota(h'') = c$
  				\State $Pub_t'(a') \gets \mathrm{seen}$
 				\State $Pub_o'(a') = Pub_o(a')$
 				\State $\sigma_c'(h'') =$ play $a'$ with probability 1
  				\State $h''a' \gets$ \Call{PubTeamConv}{$ha'$, $\mathcal G$, $\mathcal G'$}
  				\State $h'\Gamma \gets h''$ 
  			\EndFor 
  		\EndIf
  		\State \Return $h'$
	\EndFunction
\end{algorithmic}
\end{algorithm}
The algorithm recursively traverses the extensive-form game tree in a depth-first post-order fashion: for each traversed node, some corresponding nodes are instantiated in the converted game as follows. The chance, terminal, and adversary nodes are copied unaltered as the coordinator player $t$ has only access to the public information observable to the team members. 
Each team member node of the extensive form is instead mapped to a new coordinator node, in which she plays a prescription $\Gamma$ among all the combinations of possible actions for each information state $I$ belonging to the public team state. 
In other words, given a public state $S$, the coordinator issues to the players different recommendations for every possible information set belonging to $S$. 
For example, in Fig.~\ref{fig:simple_example}(a), players~1 and~2 are team members, and the decision nodes compose a unique public state, therefore there is a single information set for the coordinator player in the converted game depicted in Fig.~\ref{fig:simple_example}(b). In particular, the actions available to the coordinator are prescriptions specifying an action per information set of the extensive form (equivalently, an action per private state). See Appendix~\ref{app:info_structure} on how private information affects the construction of our conversion, and Appendix~\ref{app:figures} for a richer conversion example.

\begin{figure*}[!t]
\centering
\begin{subfigure}{.3\textwidth}
  \centering
  \includegraphics[scale=0.9]{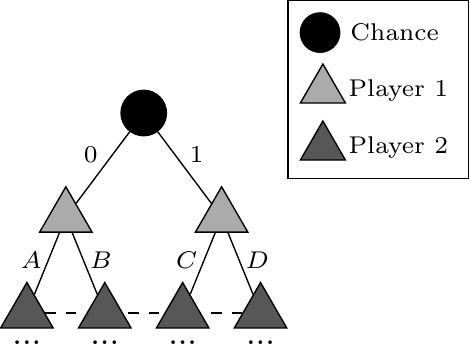}
  \caption{Example of  extensive form of an adversarial team game (only the first three levels are depicted; the symbol ``\dots''  denotes that the game continues below). Player 1 and Player 2 are in the same team and have different visibility over the chance actions.}
  \label{fig:simple_game}
\end{subfigure}%
\hspace*{0.037\textwidth}%
\begin{subfigure}{.3\textwidth}
  \centering
  \includegraphics[scale=0.9]{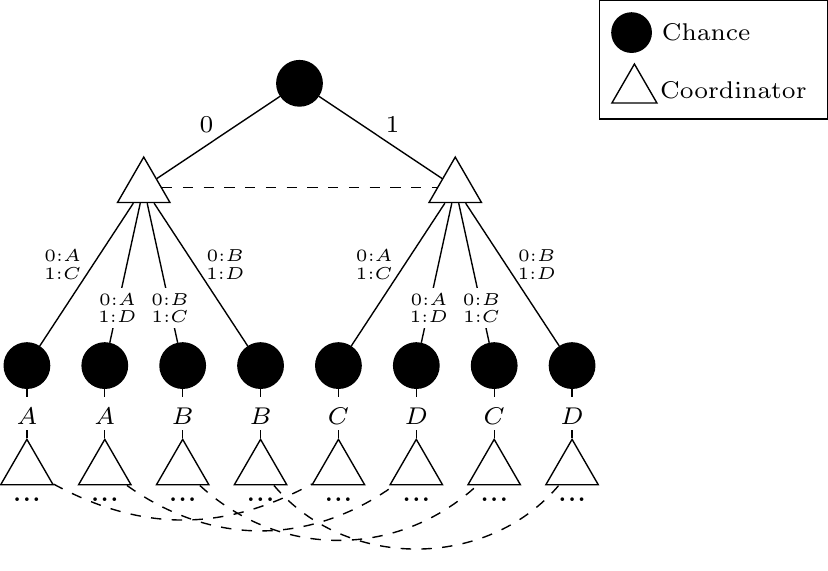}
  \caption{Team-public-information representation of the extensive-form game depicted in (a).}
  \label{fig:simple_game_converted}
\end{subfigure}
\caption{Example of game conversion: from extensive form to team-public-information representation.}
\label{fig:simple_example}
\end{figure*}

\subsection{Strategic Equivalence}

The central result of the present paper is the proof that the transformed Team Public Information game is strategically equivalent to the extensive form. In particular, we show the equivalence between a Nash Equilibrium in the converted game and the TMEcor in the extensive form. 
Before proving such a result, we introduce the following instrumental lemmas. 
We also remark that, while we make use of reduced normal-form plans, for simplicity, we refer to them as plans and pure strategies, dropping the superscript ``$^\star$''.
\begin{restatable}{lemma}{lemmaunoeq}
Given a public-turn-taking vEFG $\mathcal G$, and the corresponding TPI game $\mathcal G' = \algoname{ConvertGame}(\mathcal G)$, each joint pure strategy $\pi_{\mathcal T}$ in $\mathcal G$ can be mapped to a strategy $\pi_t$ in $\mathcal G'$, such that the traversed histories have been mapped by \algoname{PubTeamConv}. Formally, $\forall \pi_{\mathcal T}$, there is a  $\pi_t$ such that $\forall \pi_o,\pi_c$ the following holds:
\begin{equation*}
\begin{array}{c}
(\algoname{PubTeamConv}(h))_{h \text{ reached by playing } (\pi_{\mathcal T}, \pi_o, \pi_c) \text{ in } \mathcal G }\\
\equiv\\
(h')_{h'\text{ reached by playing } (\pi_t, \pi_o, \pi_c) \text{ in } \mathcal G' }. 
\end{array}
\end{equation*}
\label{th:lemma1}
\end{restatable}
\begin{restatable}{lemma}{lemmadueeq}
Given a public-turn-taking vEFG $\mathcal G$, and the corresponding TPI game $\mathcal G' = \algoname{ConvertGame}(\mathcal G)$, each coordinator pure strategy $\pi_t$ in $\mathcal G'$ can be mapped to a strategy $\pi_{\mathcal T}$ in $\mathcal G$, such that the traversed histories have been mapped by \algoname{PubTeamConv}. Formally, $\forall \pi_t$, there is $\pi_{\mathcal T}$ such that $\forall \pi_o,\pi_c$ the following holds:
\begin{equation*}
\begin{array}{c}
(\algoname{PubTeamConv}(h))_{h \text{ reached by playing } (\pi_{\mathcal T}, \pi_o, \pi_c) \text{ in } \mathcal G }\\
\equiv\\
(h')_{h' \text{ reached by playing } (\pi_t, \pi_o, \pi_c) \text{ in } \mathcal G'} .
\end{array}
\end{equation*}
\label{th:lemma2}
\end{restatable}
%
We can now define the following functions to map strategies from the extensive form game to the converted game.
\begin{definition}[Mapping functions]
We define:
\begin{itemize}
\item $\rho: \Pi_{\mathcal T} \to \Pi_t$ maps each $\pi_{\mathcal T}$ to the $\pi_t$ specified by the procedure described in the proof of Lemma~\ref{th:lemma1};
\item $\sigma: \Pi_t \to \Pi_{\mathcal T}$ maps each $\pi_t$ to the $ \pi_{\mathcal T}$ specified by the procedure described in the proof of Lemma~\ref{th:lemma2}.
\end{itemize}

Those two functions can also be extended to mixed strategies, by converting each pure plan and summing the probability masses of the converted plans. Formally, we have:
\begin{align*}
\forall \mu_{\mathcal T} \in \Delta^{\Pi_{\mathcal T}} : \rho(\mu_{\mathcal T})[\pi_t] = \sum_{\pi_{\mathcal T} : \rho(\pi_{\mathcal T})=\pi_t} \mu_{\mathcal T}(\pi_{\mathcal T}),\\
\forall \mu_t \in \Delta^{\Pi_t} : \sigma(\mu_t)[\pi_{\mathcal T}] = \sum_{\pi_t : \sigma(\pi_t)=\pi_{\mathcal T}} \mu_t(\pi_t).
\end{align*}
\end{definition}
We can now state the payoff-equivalence between a game $\mathcal{G}$ and the corresponding TPI game $\mathcal G'$ as follows. 
\begin{restatable}{theorem}{theoremequivalence}
A public-turn-taking vEFG $\mathcal G$ and its TPI game $\mathcal G' = \algoname{ConvertGame}(\mathcal G)$ are payoff-equivalent, i.e. 
\begin{align*}
\forall \pi_{\mathcal T} \; \forall \pi_o,\pi_c :  u_{\mathcal T}(\pi_{\mathcal T}, \pi_o, \pi_c) = u_t(\rho(\pi_{\mathcal T}), \pi_o, \pi_c), \\
\forall \pi_t \; \forall \pi_o,\pi_c :  u_{\mathcal T}(\sigma(\pi_t), \pi_o, \pi_c) = u_t(\pi_t, \pi_o, \pi_c).
\end{align*}
\end{restatable}
The correspondence between the strategies of the two representations is used to derive the main result of this work that shows the equivalence between a NE of the converted 2p0s game and a TMEcor of the original ATG.
\begin{restatable}{theorem}{equivalenceNE}
Given a public-turn-taking vEFG $\mathcal G$, and the corresponding TPI $\mathcal G' = \algoname{ConvertGame}(\mathcal G)$, a Nash Equilibrium $\mu_t^*$ in $\mathcal G'$ is realization equivalent to a TMEcor $\mu_\team^* = \sigma(\mu_t^*)$ in $\mathcal G$.
\label{th:main}
\end{restatable}

\subsection{Games with Compact TPI}

The procedure to convert an extensive-form game into the equivalent TPI game exploits the information structure of the team to prescribe to the team players an action for every possible private state. In general, this makes the size of the TPI to grow exponentially with the number of possible private states belonging to a public state.
However, we can find a class of games in which their information structure allows the generation of a TPI game with a size upper bounded by a polynomial in the size of the extensive form:
\begin{definition}[Games with common external information]
A vEFG $\mathcal G$ has \textit{common external information} for a set of players $\mathcal T \subseteq \mathcal N$ if all the actions performed by the other players (chance included) have the same visibility for all players in $\mathcal T$, formally, $\forall p\in\mathcal N\setminus \mathcal T,\ \forall a\in \mathcal{A}_p$: 
\begin{align*}
    Pub_{\mathcal T}(a) \neq \mathrm{priv}.
\end{align*}
\end{definition}
\begin{restatable}{theorem}{sizecommonext}
Given a public-turn-taking vEFG $\mathcal G$ with common external information for the team $\mathcal T$, the tree of corresponding TPI game $\mathcal G'$ has a number of nodes linear in the nodes of $\mathcal G$.
\label{th:common_ext}
\end{restatable}
Intuitively, Theorem~\ref{th:common_ext} states that if the game has common external information for the team, then it is possible to find the TMEcor in polynomial time. 
This result matches what was previously known in literature. When common external information is satisfied, one can, indeed, resort to \citeauthor{Kaneko1995BehaviorSM}~(\citeyear{Kaneko1995BehaviorSM}) to find a polynomial-time algorithm to find an equilibrium. This is the case, \emph{e.g.}, of Goofspiel game~\citep{ross1971goofspiel} which admits a compact TPI.

%

%
%



\subsection{TPI Expressivity and Abstractions}\label{sec:abstractions}

Abstractions demonstrated to be a successful tool to tackle real-world 2p0s game~\cite{Sandholm2015AbstractionFS}. Generally, these are obtained by merging different infosets of the same player (\emph{state} abstractions) and/or different actions of the same player at the same infoset (\emph{action} abstractions). However, despite their importance, the use of abstractions in ATGs has remained unexplored so far. 
By defining the team's strategies as behavioral, the Team-Public-Information representation provides a suitable and direct tool for designing abstractions for ATGs, while we can show that the extensive-form is not sufficiently expressive. 

\begin{restatable}{proposition}{propos} Any action or state abstraction that, once applied to an extensive-form game $\mathcal{G}$, returns a perfect-recall timeable game can be mapped specularly in the team-public-information representation $\mathcal{G}' = \algoname{ConvertGame} (\mathcal{G})$. The reverse is not true.
\label{th:prop_abs}
\end{restatable}
It can be observed that the properties required by the above proposition are satisfied by most of the abstractions, \emph{e.g.}, by~\citet{10.1145/1284320.1284324} and~\citet{DBLP:conf/aaai/GilpinSS07}.

\subsection{TPI and Subgame Solving}

Subgame solving \cite{Moravck2016RefiningSI, Brown2017SafeAN} demonstrated to be a central technique to face huge imperfect-information 2p0s games, such as, \emph{e.g.}, poker games \cite{Moravck2017DeepStackEA, BrownS17}. More precisely, subgame solving takes as input a strategy (usually called \emph{blueprint}) computed with a coarse abstraction of the game and refines it in the neighborhoods of the currently reached information set while playing (intuitively, subgame solving algorithms perform a sequence of local reoptimizations). The basic idea is to extract a portion (called \emph{subgame}) of the original game and generate on-the-fly an auxiliary game to solve just in time. The solving algorithm is initialized with the blueprint mapped to the auxiliary game and then it refines such a strategy. In particular, in every information set a player moves, the strategy refinement algorithm is performed. A notable example of subgame solving technique is \emph{depth-limited subgame solving} \cite{Brown2018DepthLimitedSF}.
In this algorithm, the auxiliary game is built starting from the subgame rooted at the public state corresponding to the infoset in which the player is playing. 
The subgame is truncated at a given depth, after which the players are assumed to play according to the blueprint. 
As widely shown in real-world applications~\cite{BrownS17}, depth-limited subgame solving can dramatically reduce the players' exploitability.

Since our TPI conversion generates a 2p0s game preserving the public structure of the original game, subgame solving techniques, including, \emph{e.g.}, depth-limited solving, can be applied directly. The only caveat concerns the size of the auxiliary game, which is exponentially larger than the size of the subgame in the extensive form. Developing efficient subgame-solving techniques for the TPI game is an interesting line of research and is left as future work. 
\section{Experimental Evaluation}\label{sec:experiments}

\begin{table*}[h!]
\centering
\caption{Size of the game trees returned by the different favours of our conversion (basic, pruned, folded, imperfect-recall abstraction of folded, and lossy imperfect-recall abstraction of folded) and by the tree decomposition by~\citet{zhang2022team}; size of the reduced normal form. We use the following notation: $mn$K$r$ is Kuhn poker with a team of $m$ players facing a team of $n$ player and $r$ ranks; $mn$L$brc$ is Leduc poker with a team of $m$ players facing a team of $n$ players, a maximum number $b$ of bets allowed in each betting round, a number of ranks $r$, and a number of indistinguishable suits $c$). Game values are provided both for the exact case (in white) and when using our abstraction (in red). The empty cells are due to instances with more than $2\cdot10^9$ nodes or out-of-memory.}
\resizebox{\columnwidth*2}{!}{%
\begin{tblr}{
    colspec = {rlllrrrrrrrrrrrr},
    rowspec={QQ|[1.5pt]QQ|QQQQQ|QQQQQ|QQQQQ|QQQQQ|QQQQ|[1.5pt]Q|[1.5pt]QQQQQQ|[1.5pt]},
}
&& & & \multicolumn{12}{c}{game instances} \\
                                                  &&          &          &  21K3  &  21K4  &  21K5  &  21K6  &  21K8  &  31K5  &  21L133  &  21L143  &  21L153  &  21L223  &  21L523  &  31L133 \\ 
\SetRow{brown9} \multirow{2}{*}{\rotatebox[origin=c]{90}{normal}}&\multirow{2}{*}{\rotatebox[origin=c]{90}{form}}& plans    & team         &  $\sim10^6$& $\sim10^8$ & $\sim10^{9}$ & $\sim10^{11}$  & $\sim10^{15}$       &    $\sim10^{20}$    & $\sim10^{70}$ &    $\sim10^{126}$      & $\sim10^{197}$         &$\sim10^{252}$& $\sim10^{3200}$         &    $\sim10^{283}$     \\
\SetRow{brown9}
& & plans    &adversary & $\sim10^{3}$ & $\sim10^{4}$ & $\sim10^{6}$ & $\sim10^{7}$       & $\sim10^{9}$       &    $\sim10^{10}$    &$\sim10^{54}$&    $\sim10^{96}$      &    $\sim10^{150}$      &$\sim10^{134}$&    $\sim10^{3900}$      & $\sim10^{152}$        \\
\multirow{5}{*}{\rotatebox[origin=c]{90}{basic}}&\multirow{5}{*}{\rotatebox[origin=c]{90}{}}& nodes    &          &7336  & 200,681 & 3,714,326 &        &        &        &35,140,264&          &          &6,140,623&          &         \\
\SetRow{gray9}                                  &  & infosets & team & 888  & 10,661 & 117,938 &        &        &      &  1,625,647&          &          &427,984&          &         \\
\SetRow{gray9}                                  &  & infosets & adversary & 12  & 16   & 20      &    &    &        &228&          &          &          630&          &         \\
&& actions  & team & 2,101 & 24,641 & 265,517 &        &       & &4,135,497&          &          &          1,287,852&          &         \\
&& actions  &  adversary & 25 & 33 & 41 &        &  &      &457&          &          &          1,443&          &         \\ 
\multirow{5}{*}{\rotatebox[origin=c]{90}{pruned}}& \multirow{5}{*}{\rotatebox[origin=c]{90}{}}& nodes    &          &   4,360     & 95,225 &  324,766 &15,007,117&        &        &35,140,264&          &          &724,009&          &         \\
\SetRow{gray9}                                   & & infosets & team &    495    &4,505&  35,943 &  267,229 & &        &101,389&          &          &45,440&          &         \\
\SetRow{gray9}                                   & & infosets & adversary & 12       &  16      &  20      &  24  &        &        &228&          &          &630&          &         \\
                                                 & & actions  & team &      1,087  & 9,849  & 77,947 & 574,709 &        &        &339,243&          &          &127,352&          &         \\
                                                &  & actions  & adversary &  25      & 33       & 41       & 49       &        &        &           457         &          &   & 1,443      &          &         \\ 
\multirow{5}{*}{\rotatebox[origin=c]{90}{folded}}&\multirow{5}{*}{\rotatebox[origin=c]{90}{}}& nodes    &          &4,108        & 66,349 & 740,406 & 7,002,763 & 488,157,721& 202,660,366 & 1,691,158 & 61,983,093 & 1,973,610,366 & 538,111 & 222,239,487 & 277,714,570\\
\SetRow{gray9}                                  &  & infosets & team & 495 & 4,505 & 35,943 & 267,229 & 13,194,833 & 11,783,620 & 96,115 & 2,625,209 & 67,400,747 & 44,252 & 18,308,851 & 17,403,080 \\
\SetRow{gray9}                                  &  & infosets & adversary & 12 & 16 & 20 & 24 & 32       & 40 & 228 & 400 & 620 & 630 & 49,584 & 816 \\
                                                &  & actions  & team & 1,086 & 9,849       & 77,947 & 574,709 & 27,978,929 &25,689,691& 208,136 & 5,736,593 & 147,671,105 &   106,963       &45,969,475&    37,743,473     \\
                                                &  & actions  & adversary & 24       & 32       & 41       & 49       & 65        &81& 457 & 801 &  1,241        & 1,443         & 123,153         &1,633\\ 
\multirow{5}{*}{\rotatebox[origin=c]{90}{imperfect-recall}}&\multirow{5}{*}{\rotatebox[origin=c]{90}{abstraction of folded}}    & nodes    &          &4,108        & 66,349 & 740,406 & 7,002,763 & 488,157,721  &202,660,366&1,691,158& 61,983,093 &1,973,610,366 & 538,111 &222,239,487&277,714,570\\
\SetRow{gray9}                                  &  & infosets & team                &   81     &    321    &1,213        & 4,585       & 68,321 &108,480& 23,071 &          4,600& 105,742 &   4,522       &361,969&184,394\\
\SetRow{gray9}                                  &  & infosets & adversary & 12         & 16 & 20 & 24 &  32      &     40     &228&           400 &   620       & 630 & 49,584 & 816\\
                                                &  & actions  & team & 253 & 1,433 & 8,237 & 48,341 & 1,710,449 & 886,591       & 13,659 & 97,577 &   682,095       & 13,646 & 1,261,733         &568,211\\
                              &  & actions  & adversary &  25      & 32       & 41       & 49       & 65 &81& 800 &         457 &     1,241               & 1,443 & 123,153 &  1,633       \\ 
 \SetRow{azure9}\multirow{4}{*}{\rotatebox[origin=c]{90}{tree}} & \multirow{4}{*}{\rotatebox[origin=c]{90}{decomposition}} & sequences    & team & 91 & 177 & not available & 433 & 801 & 2,611 & 2,725 & 6,377 & 12,361 & 5,765 & 492,605 & 42,361 \\
 \SetRow{azure9}                             &  & sequences    & adversary & 25 & 33 & not available & 49 & 65 & 81 & 457 & 801 & 1,241 & 1,433 & 123,143 & 1,633 \\
  \SetRow{azure8}                                &   & loc. feas. sets & team & 351 & 1,749 & not available & 52,669 & 1,777,061 & 974,470 & 17,718 & 115,281 & 757,884 & 21,729 & 2,042,641 & 703,390 \\
 \SetRow{azure8}                               &  & loc. feas. sets    & adversary & 25 & 33 & not available & 49 & 65 & 81 & 703 & 1,225 & 1,891 & 3,123 & 305,835 & 2,479 \\
%
%
\multicolumn{2}{r}{\textbf{exact}}  & \textbf{game value} &  &  \textbf{0.000} & \textbf{-0.0416}        & \textbf{-0.0251} & \textbf{-0.0236} & & \textbf{-0.0392}           &\textbf{0.2148}&
\textbf{0.1072}  & \textbf{0.0240} &\textbf{0.5155}&\textbf{0.9520}&\textbf{0.1894}\\
\SetRow{red9} \multirow{6}{*}{\rotatebox[origin=c]{90}{lossy imperfect-recall}}& \multirow{6}{*}{\rotatebox[origin=c]{90}{abstraction of folded}}& \textbf{game value} &  &  \textbf{-0.166} & \textbf{-0.0450} & \textbf{-0.0271} & \textbf{-0.0262} &  & \textbf{-0.0392} & \textbf{0.0888} & \textbf{0.0623} & \textbf{0.0004}  & \textbf{0.3642} & \textbf{0.5858} & \textbf{0.1894}\\
%
%
\SetRow{red8}  &  & nodes    & & 1,480 & 36,429 & 512,766 & 5,574,547 & 445,611,353 &92,309,616& 184,729       & 7,502,765 & 298,052,671 & 36,269 &3,073,197&           7,203,775 \\
\SetRow{red9}                                  &  & infosets & team &  64 & 287        & 1,146 & 4,453 & 67,803 &  91,021        &1,930& 11,981 & 70,636 & 2.513&198,329& 37,435    \\
\SetRow{red9}                                  &  & infosets & adversary & 12         & 16 & 20 & 24 & 32 &  40        &228& 400 & 620 & 630 &49,584& 816    \\
\SetRow{red8}                                    &              & actions  & team &    145    &   899     &   5,721     & 37,231& 1,517,163       &   518,591     &3,913&       30,263   & 281,981 & 5,759 &492,599&75,499\\
\SetRow{red8}                               &    & actions & adversary &   25     & 33 & 41 & 49 & 65 &   81     &457& 801 & 1,241 & 1,443 &123,153&   1,633      \\
\end{tblr}
}
\label{tab:size}
\end{table*}

\subsection{Experimental Setting}

\textbf{Game Instances}. We conduct our experimental activity with a subset of instances customarily adopted as testbed for adversarial team games, \emph{e.g.}, by~\citet{zhang2022team}. More precisely, we use multi-player parametric versions of Kuhn~\cite{kuhn1950simplified} and Leduc~\cite{southey2005bayes} poker where one player is the adversary and the remaining players collude against him. We use the following values for the parameters.
    In \emph{Kuhn} poker,  team members are from 2 to 3, ranks are from 3 to 6.
    In \emph{Leduc} poker, team members are from 2 to 3, the maximum number of bets allowed in each betting round is from 1 to 5, ranks are from 2 to 5, suits are 3.
Details are provided in Appendix~\ref{app:exp}.

\textbf{Representations}.
By exploiting the interpretability of our representation, we design pruning and/or abstraction techniques reducing the tree size. In our experiments, we focus on the following reduced representations (more details on the conversions are in Appendix~\ref{app:pruningtechniques}, while  Appendix~\ref{app:figures} provides a conversion example per representation).
    
    \emph{Basic}: it is the game returned by Algorithm~\ref{alg:PubTeam}.
    
    \emph{Pruned}: 
    The play of a public action by a team member allows to prune, in the following part of the tree, the private states with a different recommendation. 
    Thus, we safely discard a subset of the private states reducing the number of possible prescriptions in subsequent nodes. The pseudocode is in Algorithm~\ref{alg:pruned_PubTeam} in Appendix~\ref{app:pruningtechniques}. 
    %
    
    \emph{Folded}: while pruned representation allows to safely reduce the number of possible private states, it does not address the large number of nodes in the converted game. This is due to the fact that Algorithm~\ref{alg:PubTeam} preserves the chance sampling as in the original game. However, we can avoid to sample a private state and instead keep a belief over the private states of the team members.
    
    \emph{Imperfect-recall abstraction of folded}: the folded representation may include multiple replicas of the same subgames reachable from different histories. We connect the corresponding infosets in the subgames over all the replicas, thus leading to an imperfect-recall game that is \emph{well-formed} in the sense of \citet{DBLP:conf/icml/LanctotGBB12}. 
    
    \emph{Lossy imperfect-recall abstraction of folded}: we discard all coordinator's prescriptions recommending the same action (Fold or Raise or Call) to every private state. The resulting game keeps to be well-formed. 

\textbf{Algorithms}. 
We test our representations with state-of-the-art no-regret algorithms for 2p0s games as \emph{Counter Factual Regret plus} (CFR+)~\cite{Tammelin2014SolvingLI} and \emph{Outcome Sampling Monte Carlo Counter Factual Regret} (OS-MC-CFR)~\cite{Lanctot2009MonteCS}. We recall that, as showed by~\citet{DBLP:conf/icml/LanctotGBB12}, CFR-based algorithms converge to the equilibrium even with imperfect-recall games satisfying well-formed properties as for the case of our representations. To abstract from the specific implementation details, we use OpenSpiel~\cite{DBLP:journals/corr/abs-1908-09453}.

\subsection{Experimental Results}

\textbf{Representation Size and Game Value}. In Tab.~\ref{tab:size}, we report the size of the game instances obtained by our conversions, and we compare them with the size of the representation used by~\citet{zhang2022team}. Although it is not based on a tree, there is a strict connection between their representation and ours. In particular, their \emph{locally feasible sets} 
are strictly related to our actions, as they are two different approaches to describe the Cartesian product of the team members' actions given their possible private states. 
Both locally feasible sets and actions determine the size of the two representations and are helpful to analyze how their sizes grow as the size of the extensive form increases.

Interestingly, our basic representation is exponentially smaller than the reduced normal form. Furthermore, our information-lossless general-purpose techniques allow a dramatic reduction of the size of the tree up to 3 orders of magnitude. Furthermore, by using the imperfect-recall abstraction of the folded representation, we  obtain a number of actions smaller than the number of locally feasible sets, suggesting that our representation is more efficient than that by~\citet{zhang2022team}, while guaranteeing explainability and the possibility of designing abstractions. In particular, in some instances (\emph{e.g.}, 21L523), the number of actions in our representation is almost the half than the locally feasible sets. We also observe that our lossy imperfect-recall abstraction of the folded representation dramatically reduces the game size suffering from a small loss in terms of game value, averagely, $0.086$.

\textbf{Exploitability vs.~Iterations/Running Time}. 
 We show in Fig.~\ref{fig:exploitability} the dependency of the exploitability with CFR+ and OS-MC-CFR on the iterations and time for  instance 21L133. Considering the number of iterations, except for a negligible term, the exploitability with CFR+ is the same for all the information-lossless representations, while the convergence of the lossy abstraction is slightly faster. However, considering the execution time, we can fully appreciate the importance of developing techniques to reduce the representation size. Indeed, CFR+ applied to our lossy abstraction is more than one order of magnitude faster than the other representations, and even three order of magnitude faster than the basic one. This is due to the need for performing full traversals of the tree at every iteration.
 At the same time, a trajectory sampling algorithm like OS-MC-CFR benefits when reducing the number of infosets, as the variance of the estimates on the regret reduces. Remarkably, the adoption of abstractions unlocks a significant scale-up the algorithms in practice.
 
 Finally, we remark that we cannot directly compare the running time of our algorithms with that by~\citet{zhang2022team} due to the use of different technologies and implementation details. 
 Notably, our approach and that by~\citet{zhang2022team} take in input representations whose size increases with the same dependency in the size of the extensive form, suggesting that, abstracting from implementation details, the relative perfomances of these two approaches are similar to those of no-regret learning and LP with 2p0s games, see, \emph{e.g.}, \citet{DBLP:conf/icml/ZhangS20}. We point the reader to  Appendix~\ref{app:design_choice} for a detailed discussion.
 
\begin{figure}
    \centering
    \includegraphics[scale=0.47]{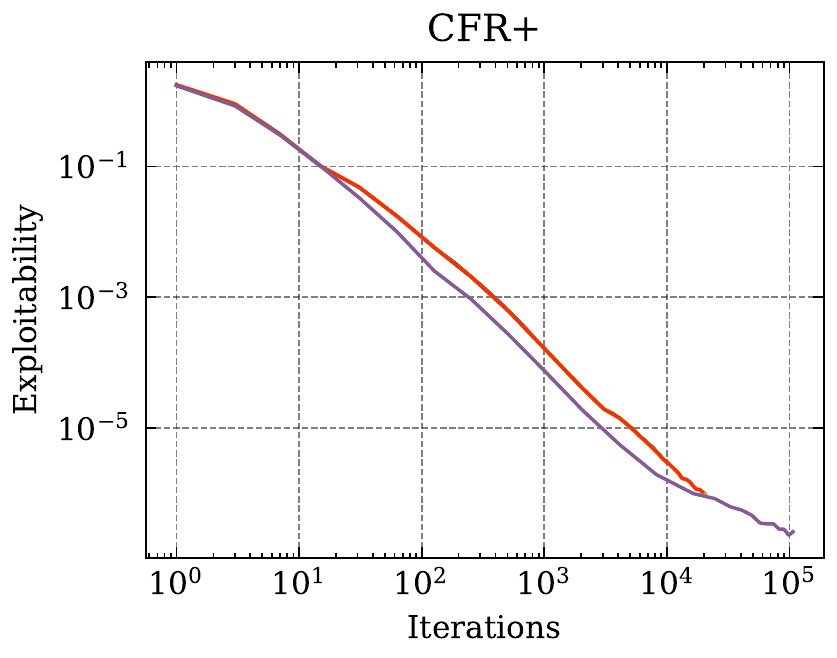}
    \includegraphics[scale=0.47]{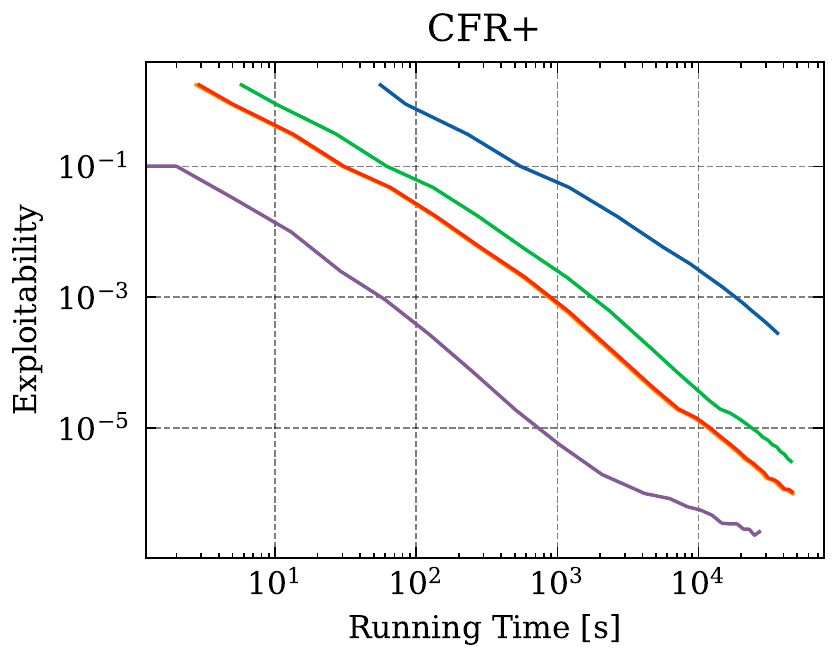}
    \includegraphics[scale=0.47]{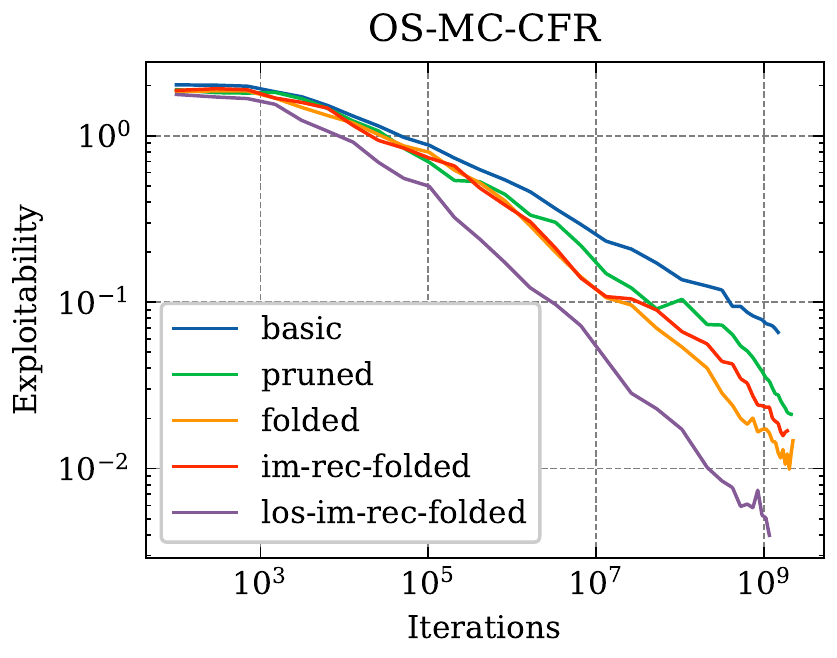}
    \includegraphics[scale=0.47]{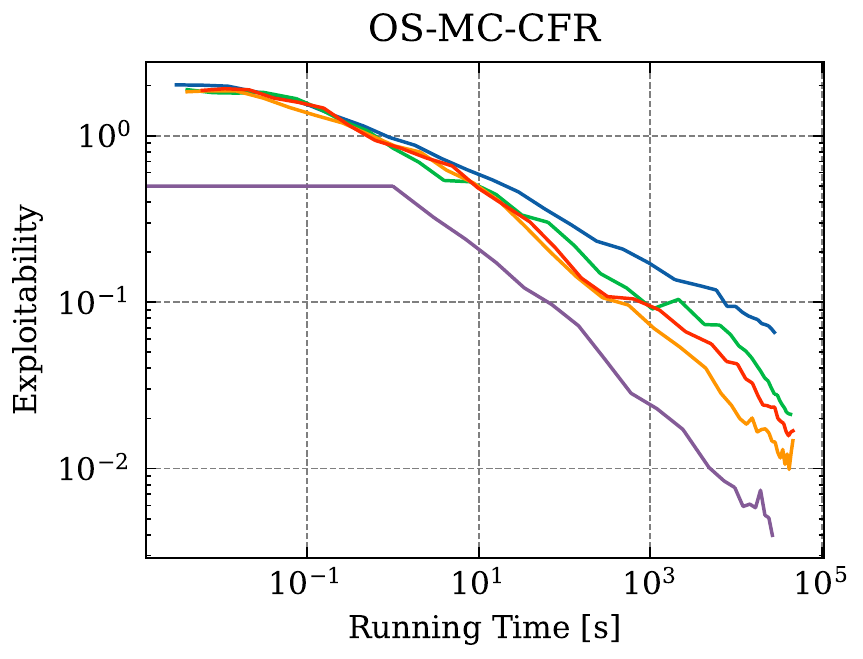}
    \caption{Exploitability of CFR+ and OS-MC-CFR with 21L133 game in the number of iterations and time (seconds).}
    \label{fig:exploitability}
    \vspace{-0.5cm}
\end{figure}
\section{Conclusions and Future Work}\label{sec:conclusion}
We bridge the realm of sequential 2-player zero-sum games with that of adversarial team games. In particular, we show that any sequential adversarial team game satisfying mild assumptions can be converted into a suitable sequential 2-player zero-sum game such that a Nash Equilibrium in the converted game is strategically equivalent to a TMEcor in the original game. This equivalence enables the adoption of successful tools for solving huge 2-player zero-sum games to adversarial team games. Furthermore, thanks to the high explainability of our representation, pruning and abstraction techniques can be easily designed to dramatically reduce the size of the tree. In particular, we empirically show that we can produce a game representation smaller than that provided by the current state of the art without any loss of information while guaranteeing explainability. Furthermore, we provide, to the best of our knowledge, the first example of abstractions for adversarial team games, showing that it allows a remarkable reduction of the tree size suffering from a small loss, and the first attempt to use no-regret learning with this class of games.
Open challenges include the design of \emph{ad hoc} algorithms for abstractions, no-regret learning, and subgame solving (whose potential impact needs to be evaluated) capable of exploiting the structure of these games to scale up to huge instances.

\bibliographystyle{icml2022}
\bibliography{example_paper}

\clearpage
\onecolumn
\appendix
\section{Proofs Omitted from the Main Paper} \label{app:proofs}
\sizepublicturn*
\begin{proof}
We provide the following procedure which returns in output a public-turn-taking game. This is achieved by assigning each level of the converted game to a player, alternating between them (chance included). Then, we add all the histories of the original game one by one, while forcing that at each level only the player corresponding to that level can play. If the history has no action assigned to the level's player, then we can add a dummy player node, with only a single action, and try to prosecute with the actions of the original history in the next node. The visibility of the added action is ``unseen'' for all players except the one playing it.

This procedure guarantees to get a strategically equivalent game by adding at most $\mathcal O((|\mathcal N| + 1)|\mathcal H|)$ for any of the $|\mathcal H|$ histories in the original game. This proves that the number of histories in the converted game is $\mathcal O((|\mathcal N| + 1)|\mathcal H|^2)$.
\end{proof}

\lemmaunoeq*
\begin{proof}
To show that \ref{th:lemma1} holds, we show how for any pure joint strategy $\pi_\team$ for the team in $\mathcal G$ it is possible to construct an equivalent pure strategy $\pi_t$ in $\mathcal G'$. 
Such goal can be achieved by recursively by traversing both $\mathcal G$ and $\mathcal G'$ while constructing $\pi_t$. 

First, consider the empty histories $h_\emptyset$ and $h_\emptyset'$ for which it trivially holds that $h_\emptyset' = \algoname{PuBTeamConv}(h_\emptyset)$.

Let $h$ and $h'=\text{PubTeamConv}(h, \mathcal G, \mathcal G')$ be the nodes currently reached by the algorithm $\algoname{PubTeamConv}$ respectively in $\mathcal G$ and $\mathcal G'$. We thus have the guarantee that $h$ and $h'$ are both terminal or both share the same player (thanks to public turn taking).
Therefore, we can differentiate between the following cases:
\begin{itemize}
\item Case \textbf{team member node}

Let $a = \pi_{\mathcal T}[ I (h)]$ be the action specified by $\pi_{\mathcal T}$ to be taken at $I(h)$. We can construct a prescription $\Gamma = (\pi_{\mathcal T}[I])_{I \in \mathcal S[h]}$ equivalent to the pure strategy $\pi_{\mathcal T}$ in this public state. We set $\pi_t[I'(h')] = \Gamma$, and prosecute our proof from the two reached nodes $h'\Gamma $ and $ha$. The construction procedure \algoname{PubTeamConv} guarantees in fact that $h'\Gamma a = \algoname{PubTeamConv}(ha)$.

\item Case \textbf{chance or opponent node}

$\pi_o$ and $\pi_c$ are common to both the traversals. This guarantees that the action $a$ suggested by the policy is equal, and by construction of the conversion procedure $h'a' = \algoname{PubTeamConv}(ha)$. We can thus proceed considering $h'a$ and $ha$.

\item Case \textbf{terminal node}

By construction, they have the same value for all players. 
\end{itemize}
This concludes the proof.
\end{proof}

\lemmadueeq*

\begin{proof}
We can prove Lemma~\ref{th:lemma2} recursively by traversing both $\mathcal G'$ and $\mathcal G$ while constructing the equivalent pure strategy in the original game. We start by $h_\emptyset'$ and $h_\emptyset$. We know that $h_\emptyset' = \algoname{PubTeamConv}(h_\emptyset)$.

As in the proof of Lemma~\ref{th:lemma1} let $h$ and $h'=\text{PubTeamConv}(h, \mathcal G, \mathcal G')$ be the nodes currently reached by the algorithm $\algoname{PubTeamConv}$ respectively in $\mathcal G$ and $\mathcal G'$. We thus have the guarantee that $h$ and $h'$ are both terminal or both share the same player (thanks to public turn taking).
Hence, we can differentiate between the following cases:
\begin{itemize}
\item Case \textbf{team member node}

Let $\Gamma = \pi_t[I'(h')]]$ be the prescription specified by $\pi_t$ to be taken at $I'(h')$. We can extract the prescribed action $a = \Gamma[I]$ to be played in history $h$. We set $\pi_{\mathcal T}[I(h)] = a$, and prosecute our proof from the two reached nodes $h'\Gamma$ and $ha$. The \algoname{PubTeamConv} procedure guarantees, indeed, that $h'\Gamma = \algoname{PubTeamConv}(ha)$.

\item Case \textbf{chance or opponent node}

$\pi_o$ and $\pi_c$ are common to both the traversals. This guarantees that the action $a$ suggested by the policy is equal, and by construction of the conversion procedure $h'a' = \algoname{PubTeamConv}(ha)$. We can thus proceed with the proof considering $h'a$ and $ha$.

\item Case \textbf{terminal node}

By construction, they have the same value for all players.
\end{itemize}
This concludes the proof.
\end{proof}
\theoremequivalence*
\begin{proof}
The proof follows trivially from Lemmas~\ref{th:lemma1}~and~\ref{th:lemma2}. Indeed, one can resort to the proof of Lemma~\ref{th:lemma1} to obtain, for each strategy $\pi_\team$ in $\mathcal G$, a payoff-equivalent strategy $\pi_t$ in $\mathcal G'$. The other direction can be obtained by following the proof of Lemma~\ref{th:lemma2}.  
\end{proof}
\equivalenceNE*
\begin{proof}
By hypothesis that $\mu_t^*$ is a NE, we have that:
\begin{equation*}
\mu_t^* \in \arg\max_{\mu_t \in \Delta^{\Pi_t}} \min_{\mu_o \in \Delta^{\Pi_o}} \sum_{\substack{\pi_t \in \Pi_t \\ \pi_o\in \Pi_o \\ \pi_c \in \Pi_c}} \mu_t(\pi_t) \mu_o(\pi_o) \mu_c(\pi_c) u_t(\pi_t, \pi_o, \pi_c).
\end{equation*}
We need to prove:
\begin{equation*}
\sigma(\mu_t^*) \in \arg\max_{\mu_{\mathcal T} \in \Delta^{\Pi_{\mathcal T}}} \min_{\mu_o \in \Delta^{\Pi_o}} \sum_{\substack{\pi_{\mathcal T} \in \Pi_{\mathcal T} \\ \pi_o\in \Pi_o \\ \pi_c \in \Pi_c}} \mu_{\mathcal T}(\pi_{\mathcal T}) \mu_o(\pi_o) \mu_c(\pi_c) u_{\mathcal T}(\pi_{\mathcal T}, \pi_o, \pi_c)
\end{equation*}

Let $\min_{TMEcor}(\mu_{\mathcal T})$ and $\min_{NE}(\mu_t)$ be the inner minimization problem in the TMECor and NE definition respectively.

\textit{Absurd.} Suppose $\exists \; \bar \mu_{\mathcal T}$ with a greater value than $\sigma(\mu_t^*)$. Formally:
\begin{equation*}
\min_{TMEcor}(\bar \mu_{\mathcal T}) > \min_{TMEcor}(\mu_t^*).
\end{equation*}
In such a case, we could define $\bar \mu_t = \rho(\bar \mu_{\mathcal T})$ having value:
\begin{equation*}
\min_{NE}(\bar \mu_t) = \min_{TMEcor}(\bar \mu_{\mathcal T}) > \min_{TMEcor}(\sigma(\mu_t^*)) = \min_{NE}(\mu_t^*),
\end{equation*}
where the equalities are due to the payoff equivalence. However this is absurd since by hypothesis $\mu_t^*$ is a maximum. Therefore necessarily:
\begin{equation*}
\sigma(\mu_t^*) \in \arg\max_{\mu_{\mathcal T} \in \Delta^{\Pi_{\mathcal T}}} \min_{NE}(\mu_{\mathcal T}).
\end{equation*}
This concludes the proof.
\end{proof}

\sizecommonext*
\begin{proof}
Consider first the opponent and chance nodes. Such nodes are copied unaltered, hence this operation does not increase the total number of nodes. 
Now, let us focus on team players' nodes. In order to prove the Theorem we have to show that, for any $h'\in\mathcal H'$, only one infoset of the original game can be mapped to the public state $\mathcal{S}_t(h')$. This ensures that node $h'$ has the same number of actions in output as infoset to which it is mapped, hence the overall number of nodes does not increase. 

Fix a node $h\in\mathcal H$ and let $h' = \text{PubTeamConv}(h, \mathcal{G}, \mathcal{G}')$. The public state is characterized by all the actions publicly observed by the team. Formally, the set of such actions in $\mathcal{G}$ at history $h$ is: 
\begin{equation*}
    \Lambda = \left\{a\in h\mid  Pub_\team(a) = \text{pub}\right\}.
\end{equation*}
Assume now, by absurd, that the set $\mathcal{S}_\team(h)$ contains two distinct information sets $I,J\in\I$. This would mean that there exists $a_I\in h_I\setminus\Lambda, a_J\in h_J\setminus\Lambda$ for $h_I\in I$, $h_J\in J$ such that:
\begin{equation}
  Pub_{p}(a_I) \neq Pub_p(a_J),
  \label{eq:info}
\end{equation}
where $p=\iota(h)\in\team$. Intuitively, the condition expressed by Equation~\eqref{eq:info} states that the two infosets are distinct.

However, this is impossible as the condition violates the assumption of A-loss refinement and common external information. This results in generating a node $h'$ with the same number of actions as $h$, hence the dimension of the TPI $\mathcal{G}'$ does not increase with respect to the dimension of $\mathcal{G}$.
\end{proof}

\propos*
\begin{proof}
Trivially, any aggregation of states or actions defined in the extensive form leads to a game that can be converted in the corresponding team-public-information representation by using Algorithm~\ref{alg:PubTeam}.
On the other hand, not all abstractions in the public information game can be reflected in the original one.
As an example, consider Figure~\ref{fig:simple_example}. If we perform \textit{action abstraction} in the converted game, by collapsing action ''0:A, 1:C'' onto action ''0:A, 1:D'', this abstraction cannot be remapped onto the original game.
This happens because such abstraction corresponds to a constraint on the possible strategies that Player 1 can choose since we are forbidding him to play any pure strategy that requires to play action A at infoset 0 and action C at infoset 1.
Such an abstraction does not modify the original game structure, since all A, B, C, D may be played for some specific prescription. 
\end{proof}

\section{Information Structure in Team Games} \label{app:info_structure}
The core problem of finding a TMEcor in adversarial team games resides in \textit{asymmetric visibility} since team members have a private state that does not allow creating a perfect recall joint coordination player by trivially merging the players without any modification of their information structure.

In the following, we characterize the possible types of asymmetric visibility that may cause imperfect recall for the joint player, and singularly address them.

\begin{itemize}
\item \textbf{Non-visibility over a team member's action}. If a team member plays an action hidden from another team member, the joint team player would have imperfect recall due to the forgetting of his own played actions. This source of imperfect recallness can be avoided in a TMEcor by considering the shared deterministic strategies before the game starts, thanks to \textit{ex-ante coordination}. This allows us to know a priori the exact actions played by team members in each node. Thus it is safe to apply a perfect recall refinement in the original game, which corresponds to always considering the chosen action of a team member as $\mathrm{obs}$ by other team members.

\item \textbf{Non-visible game structure}. Consider two nodes in the same information set for a player before which the other team member may have played a variable number of times, due to a chance outcome non-visible to the team member of these nodes. In this case, a perfect recall refinement is not applicable to distinguish the nodes, because it would give the joint coordinator information that is private of the current player. To solve this edge case, we require the property of public turn-taking.

\item \textbf{Private information disclosed by chance/adversary to specific team members.} It is the most complex type of non-visibility, since in a TMEcor we have no explicit communication channels through which to share information, and therefore this type of joint imperfect recall can only be addressed by considering a strategically equivalent representation of the game in which at most one of the team players has private information. 
\end{itemize}

\section{Pruning and Abstraction Techniques to Generate More Concise Representations}
\label{app:pruningtechniques}

\begin{algorithm}[!htb]
\caption{Pruned Public-Team Conversion}\label{alg:pruned_PubTeam}
\begin{algorithmic}[1]
	\Function {ConvertGame}{$\mathcal G$}
		\State initialize $\mathcal G'$ new game
		\State $\mathcal N' \gets \{t,o\}$
		\State $h'_{\emptyset} \gets$ \Call{PubTeamConv}{$h_{\emptyset}, \mathcal G, \mathcal G'$} \Comment{new game root}
		\State \Return{$\mathcal G'$}
	\EndFunction
	\item[] 
	\Function{PubTeamConv}{$h$, $\mathcal G$, $\mathcal G'$, $\boldsymbol{\mathcal X}$}
		\State initialize $h' \in \mathcal H'$
  		\If {$h \in \mathcal Z$} \Comment{terminal node}
  			\State $h' \gets h' \in \mathcal Z'$
  			\State $u'_p(h') \gets u_p(h) \quad \forall p \in \mathcal N$
  		\ElsIf {$\mathcal P(h) \in \{o, c\}$} \Comment{opponent or chance}
  			\State $\mathcal P'(h') \gets \mathcal P(h)$
  			\State $\mathcal A'(h') \gets \mathcal A(h)$
  			\If {$h$ is chance node}
  				\State $\sigma_c'(h') = \sigma_c(h)$
  			\EndIf
  			\For {$a' \in \mathcal A'(h')$}
  				\State $Pub_t'(a') \gets$ $\mathrm{obs}$ \textbf{if} $Pub_{\mathcal T}(a') = \mathrm{pub}$ \textbf{else} $\mathrm{unobs}$
  				\State $Pub_o'(a') \gets Pub_o(a')$
  				\State $h'a' \gets$ \Call{PubTeamConv}{$ha'$, $\mathcal G$, $\mathcal G'$, $\boldsymbol{\mathcal X}$}
  			\EndFor 
  		\Else \Comment{team member}
  			\State $\mathcal P'(h') = t$
  			\State $I \gets I(h)$
  			\State $\mathcal A'(h') \gets \bigtimes_{J \in \mathcal  S_{\team_I}(h) \boldsymbol{: \nexists J' \in \mathcal X \; matching \; J}} \mathcal A(I)$ \label{lst:line:prescription} \Comment{prescriptions}
  			\For {$\Gamma' \in \mathcal A'(h')$}
  				\State $Pub_t'(\Gamma') \gets \mathrm{seen}, Pub_o'(\Gamma') \gets \mathrm{unseen}$
  				\State $a' \gets \Gamma'[I(h)]$ \Comment{extract chosen action}
  				\State \boldsymbol{$\mathcal X \gets \mathcal X \cup \{J : \Gamma'(J) \neq a'\}$} \Comment{update $\mathcal X$ removing incompatible private states}
  				\State initialize $h'' \in \mathcal H'$
  				\State $\mathcal A'(h'') \gets \{a'\}$
 				\State $\mathcal P(h'') = c$
  				\State $Pub_t'(a') \gets \mathrm{seen}$
 				\State $Pub_o'(a') = Pub_o(a')$
 				\State $\sigma_c'(h'') =$ play $a'$ with probability 1
  				\State $h''a' \gets$ \Call{PubTeamConv}{$ha'$, $\mathcal G$, $\mathcal G'$, $\boldsymbol{\mathcal X}$}
  				\State $h'\Gamma \gets h''$ 
  			\EndFor 
  		\EndIf
  		\State \Return $h'$
	\EndFunction
\end{algorithmic}
\end{algorithm}

As aforementioned, in the worst case, our representation cannot have a size upper bounded by a polynomial in the size of the extensive form unless $\mathsf{P}= \mathsf{NP}$. Nevertheless, in many cases, the game tree generated by our conversion may contain redundant information, and thus it can be compressed without any loss of information. In the following, we provide different procedures to generate a much more concise team-public-information representation of an adversarial team game. 

\textbf{Pruned Representation}. 
Whenever the coordinator prescribes a team member to play an action $a$ such that $Pub_{\mathcal{T}}(a) = \mathrm{pub}$, where $\mathcal{T}$ is the team, the possible private states in which the player may be can be reduced after observing the action chosen from the given prescription, and this may also impact on the possible private states of other team members. Since the number of prescriptions depends on the number of private states, a dramatic reduction of the number of prescriptions is achieved without any loss of information.

The pseudocode of the procedure to directly generate a TPI in its pruned representation is provided in Algorithm~\ref{alg:pruned_PubTeam}. It takes as input the same vEFG as Algorithm~\ref{alg:PubTeam}. In particular, the procedure is obtained by a simple modification of Algorithm~\ref{alg:PubTeam}, adding a parameter $\mathcal X$ in \algoname{PubTeamConv} which is used to store the private states that can be excluded in the following part of the tree once played a public action. To ease the visualization, modifications to Algorithm~\ref{alg:PubTeam} are highlighted in bold.
By excluding every information set in $\mathcal X$ when building the prescription in Line~\ref{lst:line:prescription}, we can effectively prune the number of private states to be considered by the coordinator.
An example of pruned representation is provided in Figure~\ref{fig:convPrunedExample} in Appendix~\ref{app:figures}.

\textbf{Folded Representation}. In the basic TPI game produced by Algorithm~\ref{alg:PubTeam}, chance outcomes are explicitly represented in the game tree independently of the visibility of the outcomes, thus branching the game tree into different subgames according to the specific outcome. Consider the case of a chance action that can be observed by a team member and not observed by the adversary. In the converted game, such an action is not observable to any player, and therefore it can be safely postponed as long as no specific action depends on it. The folding representation takes advantage of this property to avoid sampling these types of private states. Instead, it samples an action from the prescription depending on the probability that a specific private state is present at a given point in the game, given the previous actions of all players and their current strategies. The dummy chance nodes $h''$ instantiated in Algorithm~\ref{alg:PubTeam} therefore may present different actions, each with a probability given by the sum of the probabilities of the private states for which that action has been prescribed. 

This approach can be considered as a hybrid game-specific representation between the public tree of the team and the original tree of the adversary, allowing a dramatic reduction of the size of games with private signals such as Poker.
To apply the folded representation to Kuhn and Leduc Poker, we maintain a belief over the possible joint cards assigned to the team members and perform a Bayesian update whenever new information is disclosed. In Poker, this happens by choosing a public action after a prescription and drawing a public card.
This belief can then be integrated with full history information (adversary and public card) to determine the probability of picking specific actions from a given prescription, and to evaluate the payoffs at the terminal nodes.
The information state of the coordinator is described by the full sequence of prescriptions given and public information for the team. This type of belief and reach probability are not novel as they have been introduced by \citet{Foerster2019BayesianAD} and \citet{Sokota2021SolvingCG} in cooperative multiagent RL settings.

The name \textit{Folded Representation} is inspired by the fact that trajectories with the same public actions but different private states are folded one over the other in the converted game. An example of folded representation is provided in Figure~\ref{fig:convFoldedExample} in Appendix~\ref{app:figures}.

\textbf{Imperfect-Recall Abstraction of the Folded Representation}. This representation takes advantage of the fact that subgames rooted in information states, whose current belief and public actions are the same, correspond to the same state of the original game. Therefore those subgames share the same structure and the same payoffs.

Thus, we can avoid including the full sequence of prescriptions in the information set of each player. This does not directly reduce the number of nodes, but it reduces the number of information sets, simplifying the information structure of the game. This also reduces the space requirements to represent the strategies and simplifies the information structure of the coordinator. This abstraction technique is theoretically sound and leads to a \emph{well-formed} game in the sense by \citet{Lanctot_Gibson_Burch_Zinkevich_Bowling_2012}. Therefore, in these settings, as showed by~\citet{Lanctot_Gibson_Burch_Zinkevich_Bowling_2012}, no-regret algorithms converge to the equilibrium.
In our experiments, we employ this information state refinement technique on top of the folded representation. An example of imperfect-recall abstraction of the folded representation is provided in Figure~\ref{fig:convFoldedIRExample} in Appendix~\ref{app:figures}.

\textbf{Lossy Imperfect-Recall Abstraction of the Folded Representation}. The compression techniques used for generating the pruned and folded representations have a high impact whenever the coordinator's prescription includes different actions to different private states. These actions are observable to the team members. Since different actions are played at different private states, observing an action reveals the private state, thus simplifying the part of the games following such a prescription. On the other hand, whenever the coordinator prescribes the same action to every private state, playing an action does not reveal any information. Therefore, the public state keeps having a combinatorial size. Intuitively, the higher the degree of signaling (communication), the smaller the size of the tree.

The main idea behind our lossy abstraction is to discard all the uninformative prescriptions recommending to play the same card at every private state. More precisely, in our Poker instance, we discard from the game tree all the prescriptions recommending to play Fold at every private state, and we do the same for the cases of Call and Raise. Notice that such discarding is equivalent to forcing the coordinator to play those prescriptions with zero probability. Interestingly, this abstraction cannot be defined on the extensive-form game, while it can be defined on our representation.

In particular, we discard the above coordinator's actions from the folded representation and apply the imperfect-recall abstraction described above, thus obtaining a \textit{well-formed game} as defined in~\citet{Lanctot_Gibson_Burch_Zinkevich_Bowling_2012}.
An example of imperfect-recall abstraction of the folded representation is provided in Figure~\ref{fig:convAbFoldedIRExample} in Appendix~\ref{app:figures}.

\section{Comparison among the Representations} \label{app:figures}

We provide an example of extensive-form game and of the three conversions described in the paper in Figs.~\ref{fig:example}--\ref{fig:convFoldedExample}.
To ease the visualization, we focus on a cooperative game with no adversary.

\begin{figure}[H]
\centering
\includegraphics[scale=0.6]{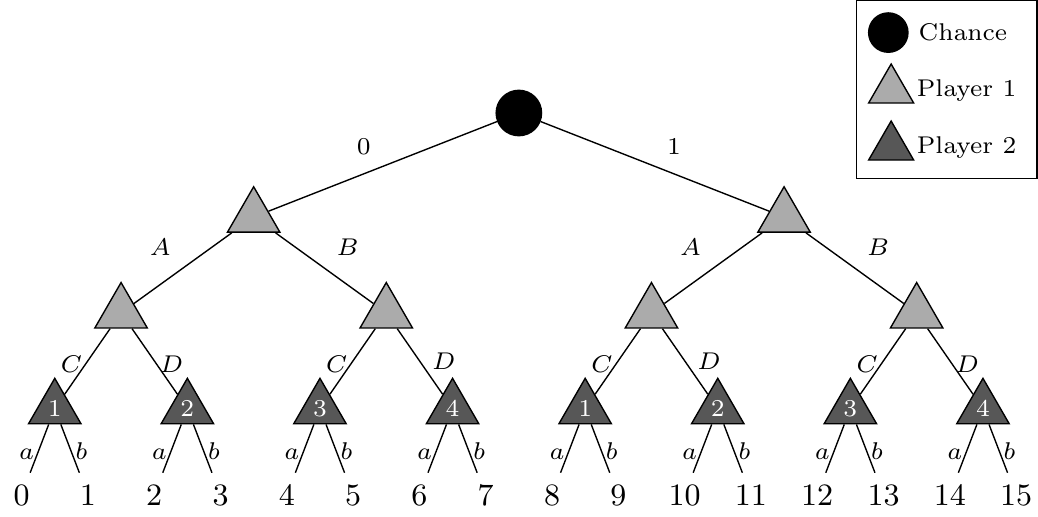}
\caption[Example of a cooperative game]{Extensive form of a 2-player team game with chance and without adversary, where Player 2 observes all actions except those of chance. Nodes of a player with same number are in the same infoset.}
\label{fig:example}
\end{figure}
\vspace{0.25cm}
\begin{figure}[H]
\centering
\includegraphics[scale=0.6]{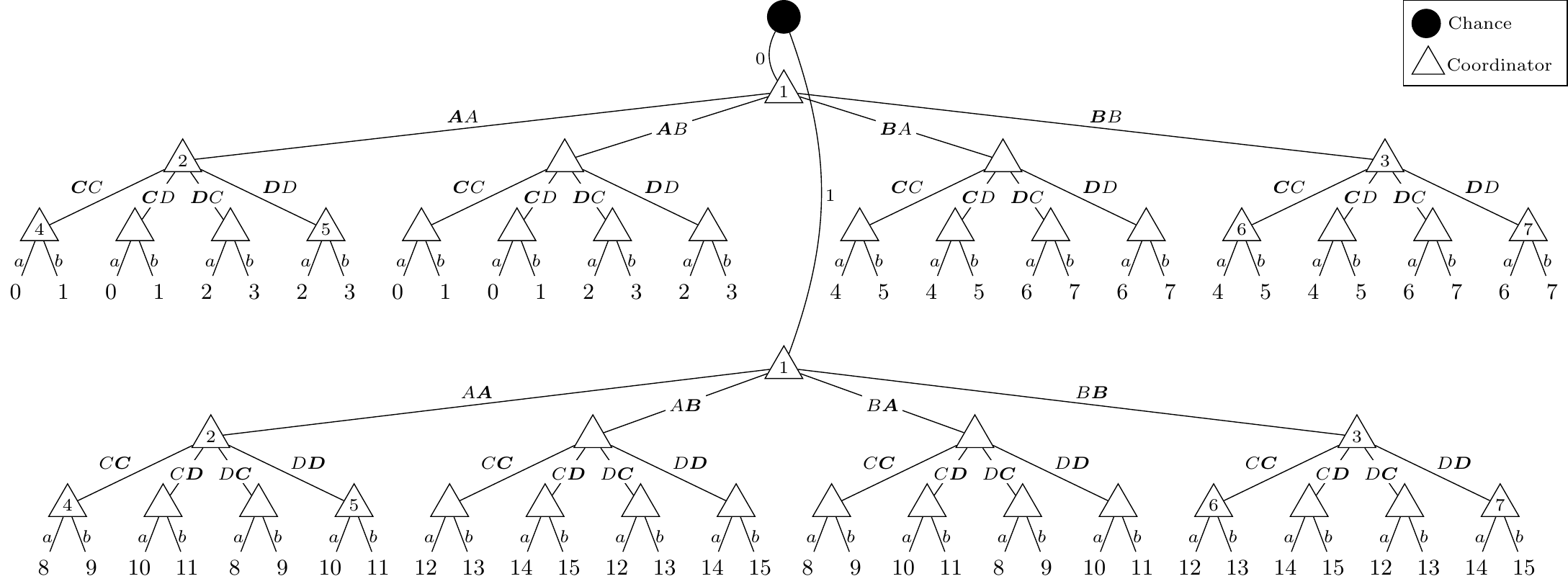}
\caption[Example of a converted game]{Team-public-information representation of the game depicted  in Figure \ref{fig:example}. Nodes of a player with same number are in the same infoset. For the sake of notation, dummy chance nodes are not represented, prescriptions include the action to take for private state $0$ and $1$, the action taken afterward is in bold in the prescription.}
\label{fig:convExample}
\end{figure}
\vspace{0.25cm}
\begin{figure}[H]
\centering
\includegraphics[scale=0.6]{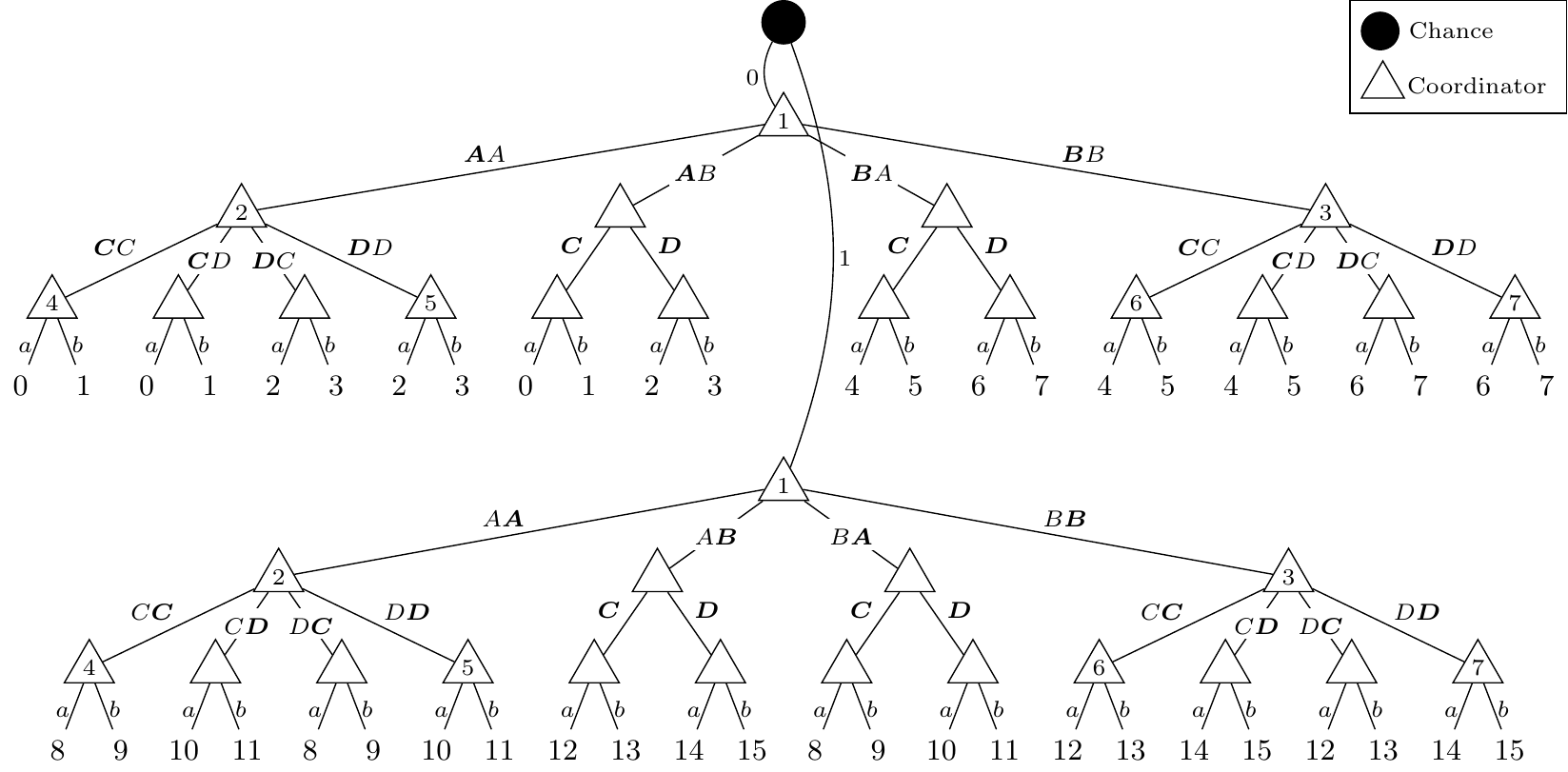}
\caption[Example of a converted pruned game]{Pruned team-public-information representation of the game depicted  in Figure \ref{fig:example}. Nodes of a player with the same number are in the same infoset. For the sake of notation, dummy chance nodes are not represented, prescriptions include the action to take for private state $0$ and $1$, the action taken afterward is in bold in the prescription.}
\label{fig:convPrunedExample}
\end{figure}
\vspace{0.25cm}
\begin{figure}[H]
\centering
\includegraphics[scale=0.6]{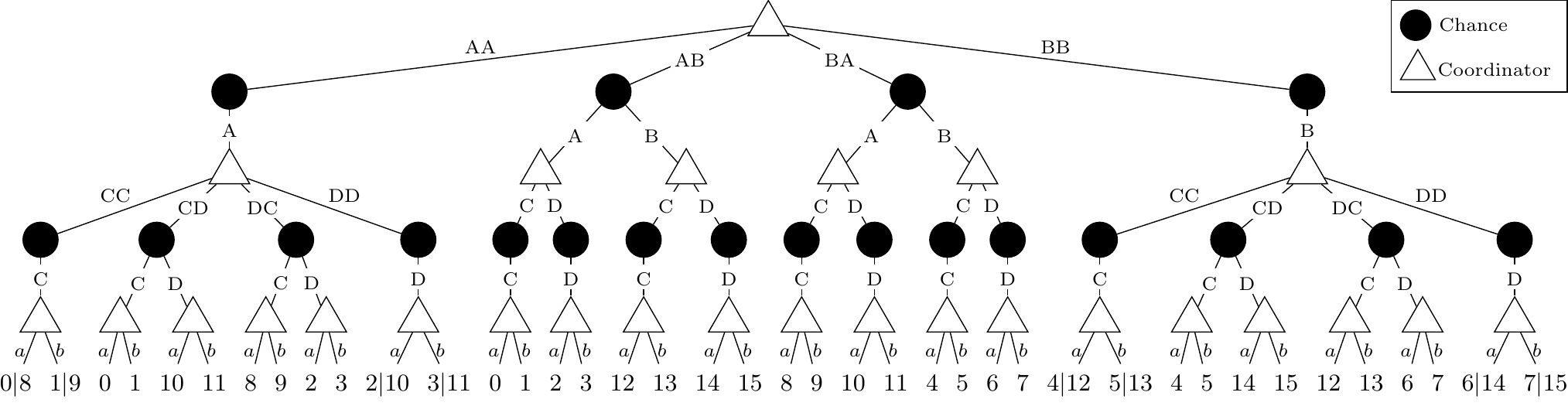}
\caption[Example of a converted folded game]{Folded team-public-information representation of the game depicted  in Figure \ref{fig:example}. For the sake of notation, prescriptions include the action to take for private state $0$ and $1$. Terminal nodes in the form $x|y$ represent a terminal node which has a weighted average value with respect to the outcomes $x$ and $y$.}
\label{fig:convFoldedExample}
\end{figure}
\vspace{0.25cm}
\begin{figure}[H]
\centering
\includegraphics[scale=0.6]{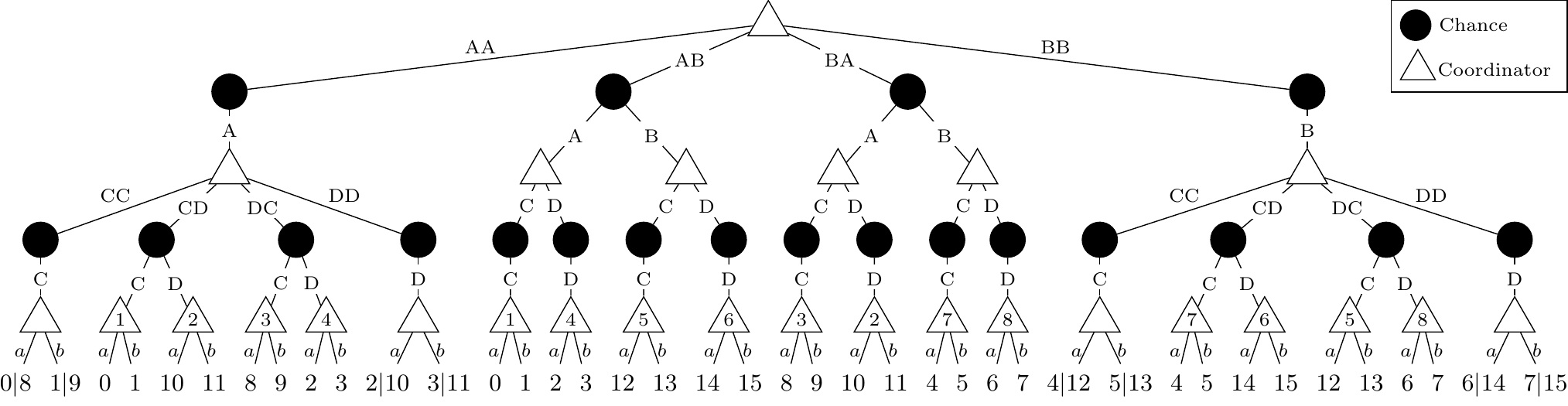}
\caption[Example of a converted folded game]{Imperfect-recall abstraction of the folded team-public-information representation of the game depicted  in Figure \ref{fig:example}. For the sake of notation, prescriptions include the action to take for private state $0$ and $1$. Terminal nodes in the form $x|y$ represent a terminal node which has a weighted average value with respect to the outcomes $x$ and $y$. Note that the coordinator has imperfect recall on the nodes characterized by the knowledge of a specific private state and sharing the same public history of played actions.}
\label{fig:convFoldedIRExample}
\end{figure}
\vspace{0.25cm}
\begin{figure}[H]
\centering
\includegraphics[scale=0.6]{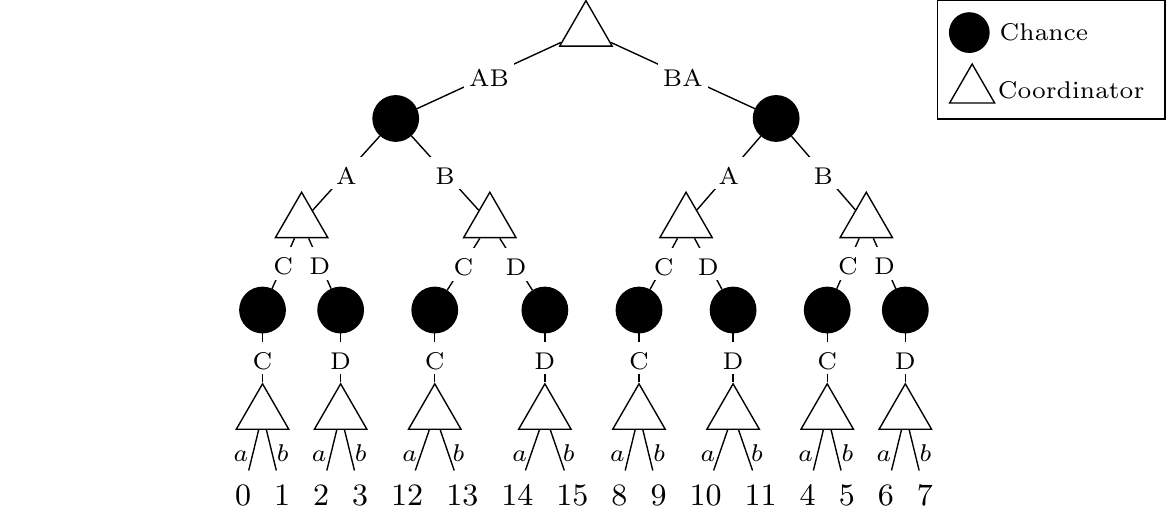}
\caption[Example of a converted folded game]{Lossy imperfect-recall abstraction of the folded team-public-information representation of the game depicted in Figure \ref{fig:example}. For the sake of notation, prescriptions include the action to take for private state $0$ and $1$. Terminal nodes in the form $x|y$ represent a terminal node which has a weighted average value with respect to the outcomes $x$ and $y$. In this case, the imperfect recall abstraction does not coarsen the information structure of the coordinator, since all the nodes at the last level are characterized by a different private state-public history combination.}
\label{fig:convAbFoldedIRExample}
\end{figure}

\FloatBarrier




\section{Experimental settings} \label{app:exp}
\subsection{Poker instances}
We refer to the three-player generalizations of Kuhn and Leduc poker proposed by \cite{Farina2018ExAC}.

Like all poker games, at the start of the game each player antes one to the pot, and receives a private card. Then players play sequentially in turn. Each player may check by adding to the pot the difference between the higher bet made by other players and their current bet (i.e. by matching the maximum bet made by others). Each player may fold whenever a check requires putting more money into the pot and the player instead decides to withdraw. Each player may raise whenever the maximum number of raises allowed by the game is not reached, by adding to the pot the amount required by a check plus an extra amount called raise amount. A betting round ends when all non-folded players except the last raising player have checked.

In \textbf{Kuhn poker}, there are three players and k possible ranks with k different ranks. The maximum number of raises is one, and the raising amount is 1. At the end of the first round, the showdown happens. The player having the highest card takes all the pot as payoff.

In \textbf{Leduc poker}, there are three players, k possible ranks having 3 cards in the deck each, and 1 or 2 raises. The raise amount is 2 for the first raise and 4 for the second raise. At the end of the first round, a public card is shown, and a new round of betting starts from the same player starting in the first round. In the end, the showdown happens. Winning players are having a private card matching the rank of the public card. If no player forms a pair, then the winning player is the one with the card with the highest rank. In the case of multiple winners, the pot is split equally.

\subsection{Implementation details}
We implemented the folded representation of both Kuhn and Leduc taking advantage of the OpenSpiel \cite{lanctot2019openspiel} framework. The framework allowed us to specify the game as an evolving state object and provided the standard resolution algorithms for the computation of a Nash Equilibrium in the converted game.

The experiments have been performed on a machine running Ubuntu 20.04 with a Intel Xeon Platinum 8358 (128) @ 3.300GHz CPU with 503 GB of memory. The implementation is single-threaded.

\subsection{Design choices}\label{app:design_choice}

 Customarily, researchers developed \emph{ad hoc} codes with different programming languages, each exploiting various programming optimization. This approach makes the comparison among the different algorithms difficult, hiding their actual scalability and sometimes emphasizing ancillary, non-central issues (\emph{e.g.}, adopting different versions of GUROBI or CPLEX). For this reason, we opted to adopt a tool publicly available to represent and solve the transformed games (\emph{i.e.}, OpenSpiel framework). While such a framework is general and readily available, some implementation choices for memory allocation and game representation slow down the performance with respect to the custom implementation by Zhang \& Sandholm (2021). The only metric allowing us to have a comparison not depending on the specific technology is the size of the optimization problem. This is the reason why we directly compare the number of variables and constraints of the linear program used by Zhang \& Sandholm (2021) with the number of infosets and actions of our game tree. Interestingly, there is a strict connection between the variables in Zhang \& Sandholm (2021) and our actions, and the number of constraints in Zhang \& Sandholm (2021) and our infosets. The interesting point is that the size of the problem by Zhang \& Sandholm (2021) and the size of problem (produced thanks to abstractions) are asymptotically the same as the size of the instance increases. This suggests that, asymptotically, the relative performance of solving our tree and the problem by Zhang \& Sandholm (2021) depend only on the two algorithms (as the size of the instances is the same). In particular, the relative performance between no-regret and linear programming is known \cite{DBLP:conf/icml/ZhangS20}.

\section{Plots}

We report a larger version of the exploitability plots provided in the main body of the paper.


\begin{figure*}[h!]
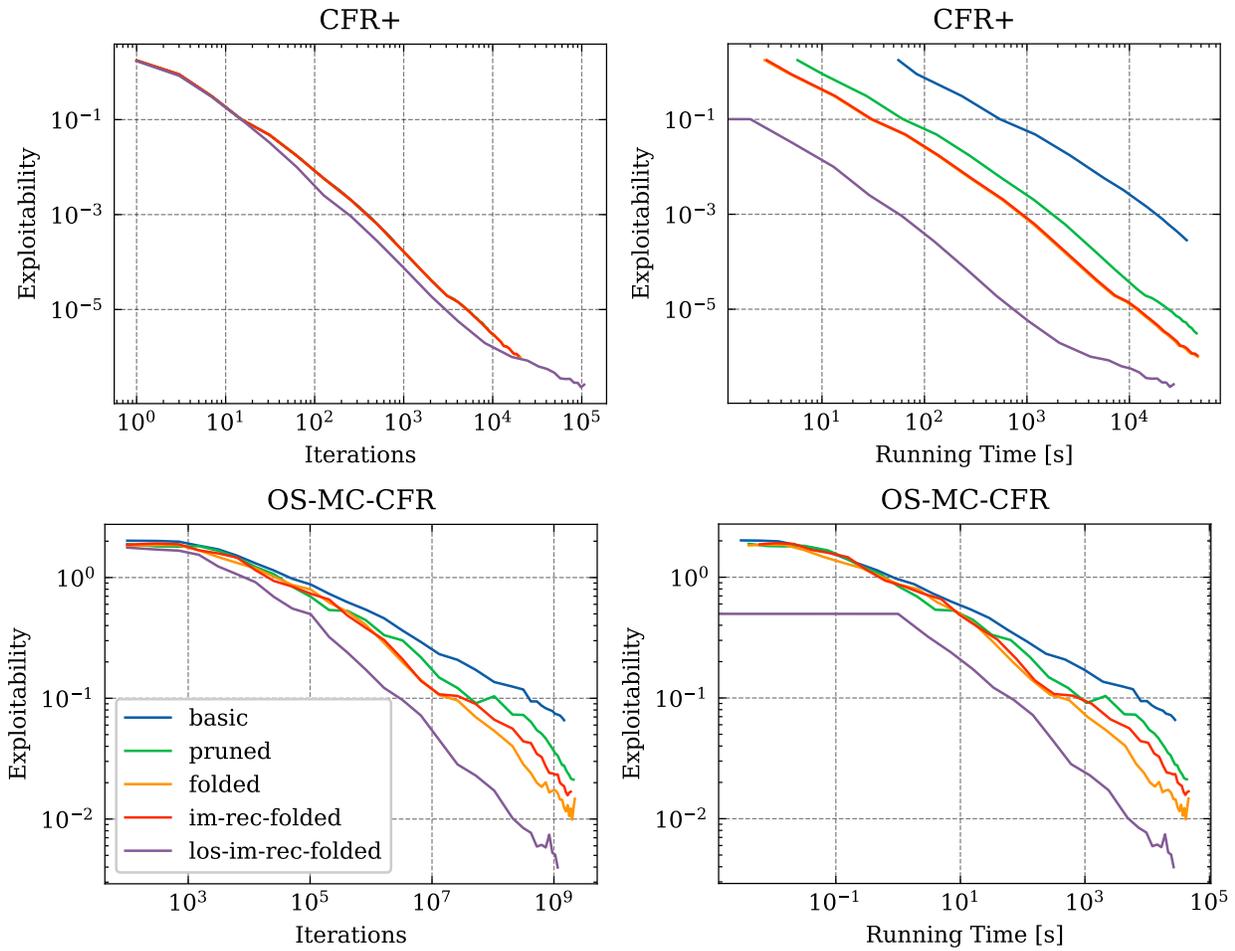

    \centering
    \includegraphics[scale=0.95]{images/CFR+_21L133Iterations.pdf}
    \includegraphics[scale=0.95]{images/CFR+_21L133Running.pdf}
    \includegraphics[scale=0.95]{images/OSMCCFR_21L133Iterations.pdf}
    \includegraphics[scale=0.95]{images/OSMCCFR_21L133Running.pdf}
    \caption{Exploitability of CFR+ and OS-MC-CFR with 21L133 game in the number of iterations and time (seconds).}
    \label{fig:exploitabilitylarge}
\end{figure*}

\end{document}